\documentclass[a4paper,11pt]{article}
\usepackage{amsthm,amsfonts,amssymb,euscript}
\usepackage{latexsym, multicol, fancybox}
\usepackage{graphicx}
\usepackage{color}
\usepackage{amsmath, amsthm, amssymb, bm}
\usepackage{epstopdf}
\usepackage{caption}
\usepackage{psfrag}
\usepackage{slashed}
\usepackage{verbatim}
\usepackage[normalem]{ulem}
\usepackage{cancel}

\usepackage{mathrsfs}
 \usepackage{xcolor}
 \usepackage[citebordercolor={green}]{hyperref}
 \usepackage{tikz}
 \usepackage{bbm}
 \usepackage{dsfont}
 \usepackage{bbold}
  \usepackage{cite}
  \usepackage{hyperref}
\usepackage{cleveref}

\numberwithin{equation}{section}

\newtheorem{theorem}{Theorem}[section]
\newtheorem{lemma}[theorem]{Lemma}
\newtheorem{proposition}[theorem]{Proposition}

\newtheorem{definition}[theorem]{Definition}
\newtheorem{remark}[theorem]{Remark}

\newcommand{\bea}{\begin{eqnarray}}
\newcommand{\eea}{\end{eqnarray}}
\newcommand{\lab}{\lababel}
\newcommand{\nn}{\nonumber}

\DeclareFontFamily{U}{mathx}{\hyphenchar\font45}
\DeclareFontShape{U}{mathx}{m}{n}{
      <5> <6> <7> <8> <9> <10>
      <10.95> <12> <14.4> <17.28> <20.74> <24.88>
      mathx10
      }{}
\DeclareSymbolFont{mathx}{U}{mathx}{m}{n}
\DeclareFontSubstitution{U}{mathx}{m}{n}
\DeclareMathAccent{\widecheck}{0}{mathx}{"71}

\def\beaa{\begin{eqnarray*}}
\def\eeaa{\end{eqnarray*}}
\def\ba{\begin{array}}
\def\ea{\end{array}}
\def\be#1{\begin{equation} \label{#1}}
\def \eeq{\end{equation}}
\def\bsplit{\begin{split}}

\def\pa{\partial}

\def\dual{{\,^*}}
\def\div{\mathrm{div}\,}

\def\hot{\widehat{\otimes}}

\def\lab{\label}
\def\f12{\frac 1 2}

\def\a{{\alpha}}

\def\be{{\beta}}
\def\ga{\gamma}
\def\Ga{\Gamma}
\def\de{\delta}

\def\eps{\varepsilon}
\def\la{\lambda}

\def\Si{\Sigma}

\def\Th{\Theta}

\def\th{\theta}
\def\vth{{\vartheta}}

\def\nab{\nabla}

\def\Up{\Upsilon}

\def\D{{\bf D}}

\def\H{{\bf H}}

\def\g{{\bf g}}

\def\k{\Th}

\def\BB{{\mathcal B}}
\def\CC{{\mathcal C}}
\def\MM{{\mathcal M}}

\def\LL{{\mathcal L}}

\def\RR{{\mathcal R}}

\def\RR{\mathcal{R}}

\def\Lie{\LL}

\def\f12{{\frac 1 2}}

\def\BB{\mathcal{B}}
\def\BBd{{}^*\mathcal{B}}
\def\KK{\mathcal{K}}
\def\KKd{{}^{*}\mathcal{K}}

\def\tr{{\mathrm{tr}}\, }

\def\Bb{\mathscr{B}}
\def\Bbd{\, ^{*  \hspace{-0.15em}}\Bb}

\def\trt{\trh \Th}

\def\Kc{\widecheck{K}}
\def\ah{\hat{a}}

\def\trth{\slashed{\mathrm{tr}}\, \th}
\def\thh{\widehat{\th}\hspace{0.05em}}
\def\thc{\widecheck{\trth}}

\def\ss{\mathfrak{s}}

\def\kh{\widehat{\k}}
\def\Thh{\widehat{\Th}}

\def\0{_0}
\def\1{}
\def\vol{\mathrm{vol}}
\def\ao{\widecheck{a}}

 \def\Rh{R\mkern-11mu /\,}
 \def\Rhh{{\widehat{\Rh}}}
 \def\nabh{\nab\mkern-12mu /\,}

\def\laph{\slashed{\Delta}}

\def\divh{\slashed{\mathrm{div}}\hspace{0.1em}}
\def\curlh{\slashed{\mathrm{curl}}\hspace{0.1em}}
\def\trh{\slashed{\mathrm{tr}}\hspace{0.1em}}

\def\Kk{\mathscr{K}}
\def\Kkd{\, ^{*  \hspace{-0.3em}}\Kk}

\def\vol{\mathrm{vol}}

\def\divz{\slashed{\mathrm{div}}^{(0)}}
\def\curlz{\slashed{\mathrm{curl}}^{(0)}}

\def\d{\slashed{\mathcal{D}}}
\def\Kc{\widecheck{K}}

\def\ah{\hat{a}}
\def\ss{\mathfrak{s}}

 \def\Rh{R\mkern-11mu /\,}
 \def\Rhh{{\widehat{\Rh}}}
 \def\nabh{\nab\mkern-12mu /\,}

\def\n{^{(n)}}
\def\nn{^{(n+1)}}
\def\nnn{^{(n+2)}}

\def\0{^{(0)}}

\def\ss{\mathfrak{s}}

\def\lapz{\laph^{\hspace{-0.1em}(0)}}
\def\nabz{\nabh^{(0)\hspace{-0.1em}}}

\def\i{^{(\infty)}}

\def\pp{p}

 \def\gz{\ga^{(0)}}
 \def\H{\mathfrak{h}}

\def\vol{\mathrm{vol}}

\def\gacr{\mathring \ga}

\begin{document}


\title{Forward Construction of Vacuum Initial Data with Borderline Decay}
\author{Xuantao Chen and Sergiu Klainerman}


\date{}
\maketitle
\abstract{We make use of the free data formalism developed in \cite{CK25}  to construct solutions of the Einstein vacuum constraint equations by integrating in the forward direction. This, together with a new gauge condition based on effective uniformization,  allows us to construct general solutions with limited decay at spacelike infinity. In particular, we construct solutions with minimal and even borderline decay, as considered in \cite{Shen23}, \cite{Shen24} in connection with the stability of the Minkowski space. In a forthcoming paper, we make use of the 
   techniques we develop here to identify and construct a general class of short-pulse Cauchy data that lead to the formation of trapped surfaces, extending the well-known result of  \cite{LiYu}.}
  
\parindent = 0 pt
\parskip = 12 pt

\tableofcontents

\section{Introduction}

In this paper, we make use of the free data formalism developed in \cite{CK25}  to construct slowly    decaying  solutions of the Einstein vacuum constraint equations
\begin{equation}\lab{ece}
    \begin{split}
        \div k-\nab\, \tr k&=0,\\
        R_g+(\tr k)^2-|k|^2&=0,
    \end{split}
\end{equation} by integrating in the forward direction. 

\subsection{Minimal and borderline decay}

An important motivation is to construct initial data for the Einstein vacuum equation that satisfies minimal and borderline decay, as  defined  in Shen \cite{Shen23,Shen24}. 
More precisely, in these works, Shen studied the future evolution of initial data with the following 
decay, where $r=\sqrt{(x^1)^2+(x^2)^2+(x^3)^2}$ is the radius in some underlying coordinates $(x^1,x^2,x^3)$: 
\begin{itemize}
\item The \textit{minimal} case, as  considered in \cite{Shen23}, for which the global stability of Minkowski is proved to be true:\footnote{The notation $O_s(r^{-\a})$ means a quantity $f$ satisfying $\pa^i f=O(r^{-\a-i})$ for $i=0,1,\cdots,s$.}
\begin{equation}\lab{eq:minimal-decay-intro}
g_{ij}-\de_{ij}=O_s( \eps_0 r^{-\de}),\quad k_{ij}=O_{s-1}(\eps_0 r^{-1-\de}).
\end{equation}
Here $\de$ is a given small positive constant.
\item The \textit{borderline}  case,  as considered in \cite{Shen24}, where an exterior stability result is proved:
\begin{equation}\lab{eq:borderline-decay-intro}
g_{ij}-\de_{ij}=O_s(\eps_0),\quad k_{ij}=O_{s-1}(\eps_0 r^{-1}).
\end{equation}
\end{itemize}
The borderline case includes, in particular, the case of \textit{weak borderline} data for which the data lies in the weighted $b$-Sobolev space\footnote{The space $H_b^{s,\a}:=\langle x\rangle^{-\a} H_b^s$, where $H_b^s$ is the $b$-Sobolev space defined by $\| u \|_{H_b^s(\mathbb{R}^3)}^2:=\sum_{i=0}^s \| \langle x\rangle^i \pa^i u\|_{L^2(\mathbb{R}^3)}^2$.} $H_b^{s,-\frac 32}\times H_b^{s-1,-\frac 12}$ for critical weight exponent, but not in $H_b^{s,-\frac 32+\de}\times H_b^{s-1,-\frac 12+\de}$ for any $\de>0$.

The construction of  minimal decaying initial data \eqref{eq:minimal-decay-intro}  can  be done   using the conformal method. Indeed, the $3D$ elliptic operator has good invertibility properties from $O(r^{-\de})$ to $O(r^{-2-\de})$ functions; see \cite{CSC79,CC81}.\footnote{Constructing data with faster decay using the conformal method requires further nontrivial work, as is done in Fang--Szeftel--Touati \cite{FST1,FST2}.} In a recent work, Shen--Wan \cite{SW26}  have used the solution operators studied in  Mao--Tao \cite{MaoTao2022} to construct
 weakly borderline initial data.
At the sup-norm level, the data of Shen--Wan behaves like 
\[ g_{ij}-\de_{ij} \sim \eps_0 L(r)^{-1},\qquad k_{ij} \sim \eps_0 r^{-1} L(r)^{-1},\]
 where $L(r)$ is a mild logarithmic loss function. This loss function provides a tiny extra room that ensures the data still lies in the space $H_b^{s,-\frac 32}\times H_b^{s-1,-\frac 12}$.

In this work, we use the method introduced in our previous work \cite{CK25} to construct initial data with minimal and borderline decay. In particular, we show the existence of a large class of data that is borderline at the level of sup-norm,  that is 
\[
g_{ij}-\de_{ij} \sim \eps_0,\quad k_{ij}\sim \eps_0 r^{-1}.
\]
 In other words,  we construct  data  $(g,k)$  which verifies \eqref{eq:borderline-decay-intro} not only as an upper bound, but also as a lower bound for some of its components; see the precise definition in Definition \ref{def:borderline-decay-data}.  Note that our result provides a large class of data, with the worst possible decay allowed in the exterior stability result   
 of \cite{Shen24}.  Our longer-term goal is to use such data to provide a counterexample to the full 
   global stability of Minkowski space for borderline data.\footnote{We conjecture that borderline data are unstable, in the context of the nonlinear stability of the Minkowski space, but do not rule out that weak borderline data, in the sense of Shen--Wan \cite{SW26},   may still lead to a global stability result.}

In our previous work \cite{CK25}, we identified gauge conditions   and   a full set of   scalars that can be viewed as the free data for solving the Einstein constraint equations. The approach   is suitable   for our purpose here,   as   it gives us the flexibility   to  prescribe   the desired  properties  for the free  data. 
We recall, however,   that the result  \cite{CK25}   was based on an integration-from-infinity approach,   suitable to construct data with much faster decay, as opposed to the case of slow   decay. 

 To deal with slowly decaying data, we are forced to modify the construction in \cite{CK25} in two different ways. The first is that we have to integrate towards infinity. The second concerns the gauge conditions. Recall that in   \cite{CK25},  we fixed a $r$-foliation by prescribing 
 two scalars  $\mu, \nu$,  as geometric conditions,  and then transported the remaining  angular variables according to the  normal flow to $r$. This fails here because,  for one reason,   the  limiting  sphere metric  at infinity  diverges.   In particular, this implies that the standard $\ell=1$ modes given by the angular coordinates are not appropriate.  To fix this problem, we modify the gauge condition   such that   the metric induced  on   each   level surface of $r$  verifies  the effective uniformization condition  used in  \cite{KS-GCM2}.  Our main theorem, see Theorem \ref{thm:main-intro},  is stated relative to this new gauge condition.

\subsubsection{Short review of \cite{CK25}}
In \cite{CK25}, we show that for prescribed scalars $(\BB,\BBd,\KK,\KKd)$ on $(r_0,\infty)\times \mathbb{S}^2$, there exists initial data $(g,k)$ solving \eqref{ece} that verifies
\begin{itemize}
    \item The gauge scalars\footnote{The scalars $\mu$ and $\nu$ address the choice of the sphere foliation on the initial data hypersurface $\Si$, as well as the way of embedding $\Si$ into a given spacetime $(\MM,\g)$. We refer to \cite[Section 2.4]{CK25} for more explanations.} $\mu_{\ell\geq 1}=\nu=0$.
    \item The prescribed scalars satisfy
    \begin{equation}\lab{eq:prescribed-scalar-CK25}
\begin{gathered}
\left(\divh Y- \BB\right)_{\ell\geq 2}= 0,\quad 
\left(\curlh Y-\dual\BB \right)_{\ell\geq 2}= 0,\\
\left(\laph \Pi-\KK \right)_{\ell\geq 2}= 0,\quad 
\left(r^{-4}(\pa_r (r^4\curlh \Xi))-\dual\KK \right)_{\ell\geq 2}= 0.
\end{gathered}
\end{equation}
\item The ADM charges can be specified at spatial infinity.
\item   Choosing appropriately the decay factors for the    defining scalars,    we  can ensure any decay   faster 
 than  $g-g_{Schw}= O(r^{-1-\de}) $, $k =O( r^{-2-\de})$.
\end{itemize}
Here  $Y_a:= R(N,e_b,e_b,e_a)$, $\Pi:= k(N, N)$, $\Xi_a:= k(N,e_a)$, with $N$ the radial normal and $\{e_a\}_{a=1,2}$ an orthonormal frame tangent to spheres, and the slashed notations $\divh$, $\curlh$, $\laph$ refer to the corresponding operator on spheres.

The key idea is that, once  the corresponding quantities are  prescribed, the Horizontal Constraint System (HCS), as defined in \cite{CK25},  is well-posed.\footnote{ That is,  well-posed as a  system of  transport  equations  coupled  with elliptic systems on the transversal foliation.    In fact,  the  choice  of  the free data  was made  to avoid loss  of derivatives in the system.}   This  procedure  can be used to construct data that decays at the rate $g-g_{Schw}=O(r^{-1-\de})$,  or  faster, by  backward  integration from infinity. 


\subsubsection{Main  difficulties in dealing with slow decay}

As mentioned above, the  method  of backward integration  fails    for slowly decaying data; we need instead to integrate  forward towards infinity.   Yet forward integration   has its own difficulties, not just in  dealing  with borderline decay,   but even in the case of  minimal   decay. We present below   the most obvious difficulties,  connected   with  the   definition and behavior of the  $\ell=1$ modes  of the HCS system. 
 In  \cite{CK25}, these modes were  defined  with respect to  a  basis  $J_p$, $p\in\{-,+,0\}$, defined  starting with the standard sphere at $r=\infty$ and transported  by $N(J_p)=0$, with $N$ the  unit normal to the $r$-foliation. This choice leads, however, to obvious divergences.

 \begin{enumerate}
 \item 
To start with, the  propagation equation  for the $\ell=1$ modes  of $K$ (Gauss curvature of the $r$-spheres),   derived in \cite[Section 2.6]{CK25}, has the form
\begin{align}\lab{eq:heuristic-K-ell=1}
\pa_r K_{\ell=1} +2r^{-1} K_{\ell=1} = \text{nonlinear terms}.
\end{align}
This  poses already a problem if we want to construct data with  the minimal decay \eqref{eq:minimal-decay-intro}, since \eqref{eq:heuristic-K-ell=1} is not consistent with the   expected  estimate $|\Kc|\lesssim \eps r^{-2-\de}$. 

\item 
A similar problem occurs  for the sphere  metric  $\ga$, which satisfies the transport equation (see   equation (2.58) in \cite{CK25})
\begin{align*}
    \slashed{\Lie}_{\ah N} (r^{-2}\ga) = 2r^{-2}\ah \thh+\ah \thc (r^{-2}\ga)+2 \ao r^{-1} (r^{-2}\ga),
\end{align*}
where $\thh$, $\thc$ are respectively the traceless part and the (linearized) trace part of the radial second fundamental form $\th$, and $\ao:=\ah-1$ is the perturbed lapse function.
Indeed,  the integration  from $r=r_0$ towards infinity either loses  a  $\de$-decay power, in the case of the minimal decay, or even yields  divergent metric coefficients, in the case of borderline decay.

\end{enumerate}

\subsection{Choice of gauge}\lab{sec:gauge-choice-intro}

\subsubsection{The effective uniformization for all spheres}

To overcome the difficulties mentioned above, we rely    on  a modified  gauge condition,      in which the choice of the  $\ell=1$ basis  
is tied to the  choice of an effective uniformization of the  $r$-spheres, inspired by  \cite{KS-GCM2}.
The result
of \cite{KS-GCM2} asserts the existence of a spherical coordinate chart $(\vth^A)$, unique up to rotation,  for which 
\begin{align}\lab{eq:K-ell=1-transport-intro}
\ga=r^2 e^{2u} \gacr,\\
\lab{eq:barycenter}
\int_{\mathbb{S}^2} e^{2u} x^i =0,
\end{align}
where $\gacr$ denotes the unit round metric under $(\vth^A)$, i.e.,
\bea
\lab{eq:standardsphere}
\gacr=d(\vth^1)^2+\sin^2\vth^1 d(\vth^2)^2.
\eea
As is used in \cite{KS-GCM2}, the condition \eqref{eq:barycenter} ensures the uniqueness of  solutions $u$  of  the equation\footnote{Recall that \eqref{eq:lap-u-K} has a kernel   spanned  by $\ell=1$ modes, at the linear level.}
\begin{equation}\lab{eq:lap-u-K}
-\laph_{\gacr} u-2u = r^2 \Kc +\text{(nonlinear terms in $\Kc$ and $u$)},\quad \Kc:=K-r^{-2}.
\end{equation}
Moreover, \eqref{eq:lap-u-K} fixes the $\ell=1$ modes of $K$, thus avoiding the problem of  having to integrate \eqref{eq:heuristic-K-ell=1}.

\subsubsection{Gauge equations}

Based on the above comments,  we  insist on making sure that  every $r$-sphere  verifies  the effective uniformization condition.\footnote{  Note that    \cite{KS-GCM2} only constructs     a fixed sphere,  denoted \textit{the intrinsic  GCM sphere}, with the effective uniformization condition. 
This sphere played a major role in  the context of Kerr stability \cite{KS:Kerr,GKS,KS-GCM1,KS-GCM2,Shen-GCMH}. For the GCM hypersurface (last slice) constructed in \cite{Shen-GCMH} (see also Theorem M3 of \cite{KS:Kerr}), the angular coordinates are transported by the normal vector $N$ to the spheres.}
In other words, we seek to write the metric in the following form
\begin{equation*}
    g=\ah^2 dr^2+ r^2 e^{2u} \gacr_{AB} (d\vth^A +b^A dr)(d\vth^B+b^B dr).
\end{equation*}
Note that  the shift $b^A=-\ah N(\vth^A)$  is generically non-zero since the choice of $\vth^A$ is made intrinsically on each sphere. 
In fact, in view of the definition of the radial second fundamental form $\th=\frac 12 \Lie_N \ga$, we have, see Section
\ref{section:gauge-eq}, under the current gauge  (where $N=\ah^{-1} (\pa_r -b^A \pa_{\vth^A})$),
\begin{equation}\lab{eq:pa-r-u-b-intro}
2(\pa_r u)e^{2u}\mathring\gamma_{AB}
=
2\ah\,r^{-2}\thh_{AB}
+
\bigl(\ah\,\trth-2r^{-1}\bigr)e^{2u}\mathring\gamma_{AB}
+
r^{-2}(\mathcal L_b\gamma)_{AB}.
\end{equation}
The traceless part of \eqref{eq:pa-r-u-b-intro} yields the following elliptic-type equation
\begin{equation}
\lab{eq:nabhotb}
\nabh\hot b =-2 \ah\, \thh.
\end{equation}
This equation has a kernel at the $\ell=1$ level. On the other hand, the trace part of \eqref{eq:pa-r-u-b-intro}, when projected to $\ell=1$, reads 
\begin{equation}\lab{eq:pa-r-u-i}
\pa_r u_{\ell=1} = (\ah\, \trth)_{\ell=1}+(\divh b)_{\ell=1}.
\end{equation}
Note that the condition \eqref{eq:barycenter} implies  the vanishing  of  $u_{\ell=1}$. Therefore, \eqref{eq:pa-r-u-i} relates $(\divh b)_{\ell=1}$ to the more   geometric quantity  $ (\ah\, \trth)_{\ell=1}$.  The  remaining ambiguity stems from  the   rotational ambiguity of  the effective uniformization conditions.
 The easiest way of removing it is to impose a vanishing condition on  the $\ell=1$ modes of $\curlh b$.\footnote{Indeed, through the integration by parts, such a condition is related to $\dual\nabh x^i$, the generators of the rotations on $\mathbb{S}^2$.}
This procedure  completely determines  $u$ and $b$.

 Recall that in \cite{CK25}, we have also  identified the following two scalars. The first one is the $\mu$, often referred to as the mass aspect function, which is defined as
\beaa
\mu:=-\laph(\log\ah)+K-\frac 14 (\trth)^2.
\eeaa
By prescribing $\mu_{\ell\geq 1}$ as zero, one relates the lapse function to $K$ and $\trth$, thus fixing this gauge freedom. The other gauge scalar is $\nu:=\divh \Xi$, where $\Xi_a:=k(N,e_a)$, which encodes the choice of ways of the embedding of the initial data hypersurface $\Si$ into a given spacetime $(\MM,\g)$.
We refer the readers to \cite[Section 2.4]{CK25} for details.
Together, these completely fix the gauge. We are now ready to proceed with   a first  statement of our main theorem.

 \subsection{First statement  of the main result}
 
 \begin{definition}\lab{def:borderline-decay-data}
 We say a set of vacuum initial data $(g,k)$ has the borderline decay, if there exists a chart $(x^i)$ in which \eqref{eq:borderline-decay-intro} holds, and, moreover, with $r=|x|$,
 \[
 \max_{i,j} \limsup_{R\to \infty}\sup_{r\geq R} \left(|g_{ij}-\de_{ij}|+|r k_{ij}|\right) >0.
\]
 \end{definition}
 
We state below the first version  of our main result, see Theorem  \ref{main-thm}   for the  detailed   version.

 \begin{theorem}[First version of main result]\lab{thm:main-intro}
 There exist  a large class of smooth    solutions to the constraint equations \eqref{ece}, parametrized by free data  as in \cite{CK25},   consistent with borderline decay, as defined in Definition \ref{def:borderline-decay-data}. 
 \end{theorem}
 
 We make some brief comments about the proof and the structure of the paper.
 \begin{itemize}
 \item This is a forward construction towards infinity from general initial conditions on a given sphere.\footnote{In particular, for our purpose here, one can simply take the initial sphere as the trivial one.} We first prove a local existence result under the set of gauge conditions set up in Section \ref{sec:gauge-choice-intro}, and extend it to $r=\infty$ through a bootstrap argument.
 \item In Section \ref{sec:local-existence}, we establish the local existence result, which includes the case with large  free scalars.\footnote{That is, the size of the free scalars $(\BB,\BBd,\KK,\KKd)$ can be any given large numbers.} This is done through an iteration argument similar to that used in  \cite{CK25},  where we also have to take into account  the presence of large quantities.\footnote{This is possible  due to the triangular structure of the HCS
 equations, see Section \ref{sec:heuristics-large-scalars}.} Moreover, as in \cite{CK25}, a limiting idenfication process is required. We need however  to modify our argument in \cite{CK25}, due to the new choice of gauge, as well as to track the effect of the presence of large quantities.
 \item In Section \ref{sec:global-result}, we establish the improved estimates needed for the bootstrap argument. In particular, the linear structure of the system allows us to integrate towards infinity. The lower bound condition is trivially satisfied once we impose the corresponding behavior of the free scalars.
 \end{itemize}

\subsection{Acknowledgements}
The authors would like to thank Jingbo Wan for many helpful comments and discussions, and Dawei Shen for explaining details on his work on the stability of the Minkowski spacetime. The first author is supported by ERC-2023 AdG 101141855 BlaHSt. The second author is supported by the NSF grant 2453843.

\section{Setup}\lab{sec:setup}
\subsection{Geometric quantities}\lab{sec:geometric-quantities}

We recall the notation of geometric quantities introduced in \cite{CK25}. Assume that $\Si:=(r_0,\infty)\times \mathbb{S}^2$ is equipped with a metric $g$ and the associated Levi--Civita connection $\nab$, and consider the foliation by the level spheres $S_r=\{r=\mathrm{const}\}$. 
Let $N$ be the unit normal to $S_r$ in $\Si$. 
We then choose an orthonormal frame $\{e_1,e_2\}$ tangent to $S_r$, so that $\{N,e_1,e_2\}$ forms an orthonormal frame on $\Si$.

 {\bf Ricci coefficients.} The corresponding Ricci (rotation) coefficients are defined by
 \begin{equation}\lab{eq:Ricci-N}
 \pp_a:=g(\nab_N N,e_a),\quad  \th_{ab}:=g(\nab_a N, e_b).
 \end{equation}
 The (sphere) trace and the traceless part of $\th$ are denoted respectively by $\trth$ and $\thh$.

{\bf Curvature components.} The curvature components are denoted as follows:
\begin{equation}\lab{eq:def-curvature-components-Si}
\slashed{R}_{ab}:=R(N,e_a,N,e_b),\quad Y_a:=R(N, e_b, e_b, e_a).
\end{equation}
 The (sphere) trace and the traceless part of $\Rh$ are denoted respectively by $\trh \Rh$ and $\Rhh$.

{\bf Components of $k$.} Given initial data $(g,k)$ and the triad $\{N,e_1,e_2\}$ on $\Si$, we define
\begin{equation*}
\Th_{ab}:=k(e_a,e_b),\quad \Xi_a:=k(N,e_a),\quad \Pi :=k(N,N).
\end{equation*}
They are well-defined scalars or sphere-tangent tensors.
 The sphere trace and the traceless part of $\k$ are denoted respectively by $\trh \k$ and $\kh$. 

{\bf The lapse function.}
The lapse function of a given $r$-foliation is defined as $\ah:=|\nab r|_g^{-1}$. It is independent of the choice of the angular coordinates $(\vth^1,\vth^2)$.
\begin{lemma}\lab{lem:loga}
We have 
\bea\lab{eq:P-loga}
\pp=-\slashed{d} (\log\ah),
\eea
where $\slashed{d}$ denotes the differential on the $r$-level spheres.
\end{lemma}
\begin{proof}
See Lemma 2.2 in \cite{CK25}.
\end{proof}

    {\bf Induced metric.}   
The induced metric on  $S_r$ is denoted by $\ga$, with the associated Levi--Civita connection $\nabh$.         
The Gauss curvature of the $r$-spheres is denoted by $K=K_\ga$.     Given standard spherical coordinates $(\vth^A)$, we also denote $\gacr$ as the standard unit round sphere metric as in \eqref{eq:standardsphere}. 

{\bf Gauge scalars.}   As in \cite{CK25}, we define
\begin{align}
\lab{eq:def-mu}
\mu& :=-\slashed{\Delta}(\log\ah)+K-\frac 14(\trth)^2, \\
\lab{eq:def-nu}
\nu & :=\divh\Xi=\de^{ab} \nabh_a \Xi_b,
\end{align}
where we recall the following definitions for two spherical tangent $1$-forms $\psi$, $\phi$, with $\in_{ab}$ denoting the area form oriented by $\in_{12}\, =1$:
\begin{equation*}
\begin{gathered}
    \psi\cdot \phi:=\delta^{ab}\psi_a\phi_b,\quad \psi\wedge\phi:=\in^{ab} \psi_a \phi_b,\quad (\psi\hot\phi)_{ab}=\psi_a\phi_b+\psi_b\phi_a-\delta_{ab}\psi\cdot\phi,\\
    \divh \psi:=\delta^{ab}\, \nabh_a \psi_b,\quad \curlh \psi:=\, \in^{ab} \nabh_a \psi_b,\quad (\nabh\hot \psi)_{ab}:=\nabh_a \psi_b+\nabh_b \psi_a-\delta_{ab}\, \divh\psi.
    \end{gathered}
\end{equation*}

{\bf Linearized quantities.} We define the following linearized quantities
\begin{equation}\lab{eq:linearized-quantities}
\thc:= \trth-2r^{-1},\quad \Kc:=K-r^{-2},\quad \ao:= \ah-1.
\end{equation}

{\bf Hodge operators.} We now recall the Hodge operators defined in \cite{CK}.
\begin{definition}
  Given a $2$-sphere $(S,\ga)$,  we denote by $\ss_0$ the set of scalar pairs, by $\ss_1$ the set of $1$-forms, and by $\ss_2$ the set of symmetric traceless covariant $2$-tensors on $S$.

We define the following Hodge operators:
    \begin{itemize}
        \item $\d_1$ takes $\ss_1$ into $\ss_0$:\qquad     $\d_1\xi :=(\divh\xi,\curlh\xi),$
        \item $\d_2$ takes $\ss_2$ into $\ss_1$:\qquad $   (\d_2 h)_a :=\nabh^b h_{ab}$.
    \end{itemize}
\end{definition}

{\bf The $\ell=1$ basis.}
Given spherical coordinates $(\vth^1,\vth^2)$, we define the following standard $\ell=1$ basis:
\begin{align*}
J_0:= \sqrt{\frac{3}{4\pi}}\cos\vth^1, \quad
 J_+:=\sqrt{\frac{3}{4\pi}}\sin\vth^1\cos\vth^2, \quad J_-:=\sqrt{\frac{3}{4\pi}}\sin\vth^1\sin\vth^2.
\end{align*}
For any scalar field $\phi$ on the sphere, one can decompose uniquely 
\bea\lab{eq:decomposition-ell-geq-2}
\phi=\phi_{\ell\leq 1}+\phi_{\ell\geq 2},
\eea
where $\phi_{\ell\leq 1}$ is spanned by $\{1,J_0,J_+,J_-\}$, and is orthogonal to $\phi_{\ell\geq 2}$ with respect to the measure induced by $r^{-2}\gz=\gacr$.

{\bf Norms.} We have the following definition of the weighted Sobolev norms on spheres.
\begin{definition}\lab{def:L2-Hs-S_r}
For $S_r$-tangent covariant rank-$k$ tensors $U_{a_1\cdots a_k}$,
we denote by $L_{\ga}^2(S_r)$ the $L^2$ space associated with the metric $\ga$, and by $\H_{\ga}^s(S_r)$ the Sobolev spaces for positive integers $s$, defined through $r\nabh$ where $\nabh$ is the covariant derivative with respect to $\ga$, i.e.,
\begin{equation*}
\| U \|_{\H_\ga^s(S_r)}:= \sum_{i\leq s} \|(r\nabh)^i U \|_{L_\ga^2(S_r)}.
\end{equation*}
The $L_\ga^\infty(S_r)$ space is defined through the norm
\begin{equation*}
\| U \|_{L_\ga^\infty(S_r)}:={\mathrm{ess} \sup}_{S_r} |\langle U,U\rangle_{\ga}|^\frac 12.
\end{equation*}
We also denote for simplicity $\H_0^s:=\H_{\gz}^s$,  with  $\gz$  the standard  sphere (see\eqref{eq:standardsphere}).
\end{definition}
Note that the notation deviates slightly from that of \cite{CK25}, where the $\H^s$ norm corresponds to the $\H_0^s$ norm here.

\subsection{The structure and Bianchi equations}
We define the constraint quantities:
\bea
\lab{eq:def-CC-Mom}
 \CC_{Mom} (g,k)&:=&\div k-\nab\, \tr k ,\\
 \lab{eq:def-CC-Ham}
\CC_{Ham} (g,k)&:=&R_{g}+(\tr k)^2-|k|^2.
\eea
We have the following equations derived in \cite{CK25}:
\begin{proposition}[Unconditional equations]\lab{prop:Unconditional-equations-1}
The following equations hold true:
\begin{align}
\lab{eq:structure-constraint-HH}
\nabh_N \trth &= \divh \pp-\frac 12 |\thh|^2 -\frac 34(\trth)^2-|\pp|^2+K +\Pi\, \trt+\frac 14 (\trt)^2 \\
&\notag\quad -|\Xi|^2-\frac 12 |\kh|^2-\frac 12 \CC_{Ham}, \\
\lab{eq:unconditional-Codazzi}
\divh \thh&=\frac 12 \nabh\trth-Y,\\
\lab{eq:unconditional-Bianchi}
\nabh_N K &= -\divh Y-\trth\, K+2\pp\cdot Y
-\thh\cdot (\nabh\hot \pp-\pp\hot \pp) \\
&\nonumber \quad +\frac 12\trth\, (\divh \pp-|\pp|^2),\\
\nabh_N \trt&=\divh \Xi + \trth \Pi -\thh\cdot\kh-\frac 12 \trth\trt -2\pp\cdot \Xi-(\CC_{Mom})_N,
\lab{eq:structure-constraint-Phi-R}\\
\nabh_N \Xi &= -\divh \kh+\pp\cdot\k-\Pi \pp-\frac 32\trth\, \Xi-\thh\cdot \Xi+\frac 12 \nabh \trt +\nabh\, \Pi+ \slashed{\CC}_{Mom}.
\end{align}
\end{proposition}

\subsection{The gauge equations}
\lab{section:gauge-eq}
In this section, we derive some identities that hold for the metric of the form\footnote{We use the letters $a$, $b$, $\cdots$ for frame indices on $S_r$, and the letters $A$, $B$, $\cdots$ for the angular coordinate indices.}
\bea\lab{eq:metric-canonical-form}
g=\ah^2 dr^2+ r^2 e^{2u} \gacr_{AB} (d\vth^A +b^A dr)(d\vth^B+b^B dr),
\eea
with $u$  and $b$ satisfying
\begin{equation}\lab{eq:barycenter-1}
\int_{\mathbb{S}^2} e^{2u} x^i =0,\quad (e^{2u} \curlh b)_{\ell=1}=0.
\end{equation}
We define the geometric quantities with respect to the $r$-foliation. 
\begin{proposition}\lab{prop:basic-relation-with-shift}
For a metric of the form \eqref{eq:metric-canonical-form}, the following relations hold:
\begin{align}
\lab{eq:nab-hot-b}
\nabh\hot b &=-2 \ah\, \thh, \\
\lab{eq:pa-r-u}
2\pa_r u &= \ah\, \trth-2r^{-1} + \divh b,\\
\lab{eq:laph-u-K}
-\laph_{\mathring\gamma}u+1 & = (r^2K)e^{2u}.
\end{align}
\end{proposition}
\begin{proof}

Note that the radial unit normal vector field can be written as $N=\ah^{-1} (\pa_r -b^A \pa_{\vth^A})$.
Since
$\mathcal L_{\ah N}\gamma = 2\ah\,\theta$,
we have
\[\pa_r\gamma_{AB}=
2\ah\,\th_{AB}+
(\mathcal L_b\ga)_{AB}.\]
Then, for the rescaled metric
$r^{-2}\ga_{AB}$,
we have
\bea\lab{eq:pa-r-ga-bar}
\pa_r(r^{-2}\ga)_{AB}=
2\ah \,r^{-2}\thh_{AB}+
\big(\ah\,\trth-2r^{-1}\big)r^{-2}\ga_{AB}+r^{-2}(\mathcal L_b\ga)_{AB}.
\eea
On the other hand, 
$r^{-2} \ga_{AB}=e^{2u}\mathring\ga_{AB}$,
which implies
\[
\pa_r(r^{-2}\ga_{AB})
=
2(\pa_r u)e^{2u}\mathring\ga_{AB}.
\]
Comparing this with \eqref{eq:pa-r-ga-bar}, we obtain
\begin{align}\lab{eq:identity-pa-r-u-Lie-b-ga}
2(\pa_r u)e^{2u}\mathring\ga_{AB}
=
2\ah\,r^{-2}\thh_{AB}
+
\big(\ah\,\trth-2r^{-1}\big)e^{2u}\mathring\ga_{AB}
+
r^{-2}(\mathcal L_b\ga)_{AB}.
\end{align}
Recall that $\Lie_b \ga_{AB}=\nabh_A b_B+\nabh_B b_A$.
Taking the traceless part of \eqref{eq:identity-pa-r-u-Lie-b-ga}, we obtain
\[
\nabh \hot b=-2\ah \, \thh.
\]
This proves \eqref{eq:nab-hot-b}.
Taking the trace of \eqref{eq:identity-pa-r-u-Lie-b-ga} instead gives
\[
2\pa_r u
=
\ah\, \trth-2r^{-1}
+
\divh b.
\]
This proves \eqref{eq:pa-r-u}.

We have $K_{r^{-2}\ga}=r^2 K$.
On the other hand, under a conformal change
$r^{-2}\ga=e^{2u}\mathring\ga$, in dimension $2$, the Gauss curvature satisfies
\[K_{r^{-2}\ga}= e^{-2u}\bigl(1-\laph_{\mathring\ga}u\bigr). \]
Hence $u$ satisfies the elliptic equation
\[
-\laph_{\mathring\ga}u+1=(r^2K)e^{2u}.
\]
This proves \eqref{eq:laph-u-K}.
\end{proof}

We introduce the following notation
\begin{equation}\lab{eq:def-b-flat-0}
b^{\flat,0}_A:=(\gz)_{AB}b^B.
\end{equation}
Under this convention, any $\ga$-horizontal operator are acting on $e^{2u} b^{\flat,0}$, e.g., $\divh b =\divh (e^{2u} b^{\flat,0})$.
Since the metric changes throughout the construction, we will systematically regard $b$ as the $1$-form $b^{\flat,0}$ obtained by lowering indices using the background metric $\gz$.
\begin{proposition}\lab{prop:b-flat}
The following relations hold:
\begin{align}
\lab{eq:nab-hot-b-flat}
\nabz\hot b^{\flat,0} &=-2 e^{-2u} \ah\, \thh, \\
\lab{eq:divz-b-flat-ell=1}
(\divz (e^{2u} b^{\flat,0}))_{\ell=1} & = -(e^{2u} (\ah\trth-2r^{-1}))_{\ell=1},\\
\lab{eq:curlz-b-flat-ell=1}
(\curlz (e^{2u} b^{\flat,0}))_{\ell=1} &=0.
\end{align}
The last two equations are implied by the conditions in \eqref{eq:barycenter-1}.
\end{proposition}
\begin{proof}
Note that we can also write $\Lie_b \ga_{AB}=e^{2u}
\left(
\nabz_A b^{\flat,0}_B
+
\nabz_B b^{\flat,0}_A
+
2b(u)(\gz)_{AB}
\right)$.
Taking the traceless part and combining it with \eqref{eq:identity-pa-r-u-Lie-b-ga}, we obtain
\[
e^{2u}(\nabz\hot b^{\flat,0})_{AB}=-2\ah \, \thh_{AB}.
\]
This proves \eqref{eq:nab-hot-b-flat}.
To prove \eqref{eq:divz-b-flat-ell=1}-\eqref{eq:curlz-b-flat-ell=1}, we first recall the following relations, see Lemma \ref{lem:div-curl-hot-identities}:
\beaa
\divh b= e^{-2u} \divz (e^{2u} b^{\flat,0}), \quad \curlh b=
e^{-2u}\curlz (e^{2u}b^{\flat,0}).
\eeaa
This immediately proves \eqref{eq:curlz-b-flat-ell=1} in view of the second condition in \eqref{eq:barycenter-1}.
Differentiating the first condition \eqref{eq:barycenter-1} in $r$ gives
\beaa
0=\int_{\mathbb{S}^2} 2(\pa_r u)e^{2u}x^i.
\eeaa
Substituting \eqref{eq:pa-r-u}, we find
\begin{equation*}
0
=
\int_{\mathbb{S}^2}
e^{2u}x^i
\left(
\ah\tr\theta+\divh b-2r^{-1}
\right)
=\int_{\mathbb{S}^2}
\divz (e^{2u} b^{\flat,0})x^i+e^{2u}x^i
\left(
\ah\tr\theta-2r^{-1}
\right).
\end{equation*}
This proves \eqref{eq:divz-b-flat-ell=1}.
\end{proof}

\subsection{The equations in perturbative form}

We use the following schematic notation:\footnote{The checked quantities are defined in \eqref{eq:linearized-quantities}.}
\bea\lab{eq:schematic-quantities}
\Ga_0=\{\ao,u\},\quad \Ga_1=\{r^{-1}\Ga_0,\thc,\, \slashed{d} (\log\ah),\trt,\Xi \},\quad \Ga_2=\{r^{-1}\Ga_1,\nabh\Ga_1,\Kc\}.
\eea
We also introduce the following enlarged sets:
\bea\lab{eq:schematic-quantities-enlarged}
\widetilde\Ga_0=\{\ao,u,b\},\quad \widetilde\Ga_1=\{r^{-1}\widetilde \Ga_0,\thc,\, \slashed{d} (\log\ah),\Xi,\trt,\thh,\Thh,\Pi \},\quad \widetilde\Ga_2=\{r^{-1}\widetilde \Ga_1,\nabh \widetilde \Ga_1,\Kc,Y\}.
\eea
\begin{remark}
\lab{remark;small-large}
Our local existence result   allows  the quantities  $\widetilde{\Ga}\setminus \Ga$ to be large, while the  $\Ga$ 
 quantities are kept small,  see also Remark   \ref{rem:large-free-scalars}. 
To prove the main borderline decay result of this paper, we  can assume  that all quantities  are small, that is we do not need to make the distinction at all.

\end{remark}

\begin{proposition}
\lab{prop:linearized-eqns}
The following equations hold:
\begin{align}
\lab{eq:R-transport-kac}
N \thc &= \mu -2r^{-1}\thc-2 r^{-2}\ao +\widetilde\Ga_1\cdot \widetilde\Ga_1-\frac 12 \CC_{Ham}, \\
\lab{eq:R-transport-Kc}
N \Kc &= r^{-1} \mu- \divh Y-3r^{-1}\Kc -2 r^{-3} \ao+\Ga_1\cdot \Ga_2\\
&\notag \quad +2\, \slashed{d}(\log\ah)\cdot Y-\thh\cdot \left(\nabh\hot\, \slashed{d}(\log\ah)-\slashed{d}(\log\ah)\hot\, \slashed{d}(\log\ah)\right) ,\\
\lab{eq:R-transport-ao}
\laph \ao&= \Kc- r^{-1} \thc-\mu -\laph(\Ga_0\cdot \Ga_0)+\Ga_1\cdot\Ga_1, 
\end{align}
\begin{align}
\lab{eq:unconditional-Codazzi-1}
\d_1 \d_2 \thh&= (\frac 12 \laph\trth,0)-(\divh Y,\curlh Y),\\
\lab{eq:N-trt-linearized}
N \trt&= \divh \Xi + 2 r^{-1} \Pi  - r^{-1} \trt +\widetilde\Ga_1\cdot \widetilde\Ga_1 - (\CC_{Mom})_N,\\
\lab{eq:N-div-Xi}
\ah N \divh\Xi &= - \divh\divh (\ah\kh)-4 r^{-1} \, \divh\Xi+\laph(\Pi+\frac 12 \trt)+\divh (\ah\, \slashed\CC_{Mom}) \\
& \quad +\Ga_1\cdot \Ga_2+ \nabh\left((\Pi,\thh) \cdot (\slashed{d}(\log\ah),\Xi)\right)+Y\cdot \Xi,\notag \\
\lab{eq:N-curl-Xi}
\ah N \curlh \Xi &= - \curlh\divh (\ah\kh)-4 r^{-1} \, \curlh\Xi+\curlh (\ah\, \slashed\CC_{Mom})\\
&\notag \quad +\Ga_1\cdot \Ga_2+ \nabh\left((\Pi,\thh) \cdot (\slashed{d}(\log\ah),\Xi)\right)+Y\cdot \Xi, 
\end{align}
\begin{align}
\lab{eq:laph-u-K-pert}
-\laph_{\mathring\ga}u-2u & = r^2 \Kc+r^2\Kc (e^{2u}-1)+(e^{2u}-1-2u), \\
\lab{eq:u-ell=1-pert}
u_{\ell=1} &= -\frac 12 (e^{2u}-1-2u)_{\ell=1},\\
\lab{eq:nab-hot-b-pert}
\nabz\hot b^{\flat,0} &=-2 e^{-2u} \ah\, \thh, \\
\lab{eq:divz-b-flat-ell=1-pert}
(\divz (e^{2u} b^{\flat,0}))_{\ell=1} & = -(e^{2u} (\ah\, \trth-2r^{-1}))_{\ell=1},\\
\lab{eq:curlz-b-flat-ell=1-pert}
(\curlz (e^{2u} b^{\flat,0}))_{\ell=1} &=0.
\end{align}
Here, the terms with $\cdot$    are schematic  expressions of    various    tensorial products.
\end{proposition}
\begin{proof}
The equations \eqref{eq:R-transport-kac}-\eqref{eq:N-curl-Xi} follow from  Proposition \ref{prop:Unconditional-equations-1} in the same way as Proposition 2.25 in \cite{CK25} (by taking $\Up=1$ there) and Lemma \ref{lem:loga}. 
For \eqref{eq:N-div-Xi}, we replace $\frac 12\ah\laph\trt$ and $\laph(\ah \Pi)$ with $\frac 12 \laph\trt$ and $\laph\Pi$, respectively, since the difference can be absorbed into $\Ga_1\cdot \Ga_2$ in the equations (note that $\laph\log\ah\in \nabh\Ga_1\subseteq\Ga_2$). 
Moreover, in \eqref{eq:R-transport-Kc}, \eqref{eq:N-div-Xi}, and \eqref{eq:N-curl-Xi}, we keep track of the products involving $\thh,\Thh,\Pi,Y$ from Proposition \ref{prop:Unconditional-equations-1}, including those arising from the commutators of $N$ with $\divh$ and $\curlh$ according to Lemma \ref{lem:commutation-1}.
The equations \eqref{eq:laph-u-K-pert}-\eqref{eq:curlz-b-flat-ell=1-pert} follow directly from Propositions \ref{prop:basic-relation-with-shift} and \ref{prop:b-flat}.
\end{proof}
We now introduce some notations for error terms.
\begin{definition}\lab{def:error-terms}
We denote by $\RR_{tran}(\phi)$, schematically, the following expressions
\begin{equation*}
N(\phi_{\ell=0})-(N\phi)_{\ell=0},\quad N(\phi_{\ell= 1})-(N\phi)_{\ell= 1},\quad N(\phi_{\ell\geq 2})-(N\phi)_{\ell\geq 2}.
\end{equation*}
We denote by $\RR_{ellip}(\phi)$, schematically, the following expressions
\begin{equation*}
(r^{-2} \laph^{-1} \phi)_{\ell\geq 2}-r^{-2} \laph^{-1} (\phi_{\ell\geq 2}),\quad (r^{-2} \laph^{-1} \phi)_{\ell=1}-r^{-2} \laph^{-1} (\phi_{\ell=1}),\quad (r^{-2}\laph^{-1} \phi)_{\ell=0}.
\end{equation*}
We also use $\RR(\phi)$ to denote errors of either type.
\end{definition}

Using such notations, we write, according to \eqref{eq:R-transport-Kc},
\begin{equation}
\lab{eq:N(Kc)}
\bsplit
N \Kc_{\ell\neq 1} &= \left(r^{-1} \mu-\divh Y-3r^{-1}\Kc -2 r^{-3} \ao+\Ga_1\cdot \Ga_2 \right)_{\ell\neq 1} \\
&\quad +2\left(\, \slashed{d}(\log\ah)\cdot Y-\thh\cdot \left(\nabh\hot\, \slashed{d}(\log\ah)-\slashed{d}(\log\ah)\hot\, \slashed{d}(\log\ah)\right)\right)_{\ell\neq 1}+ \RR_{tran} (\Kc).
\end{split}
\end{equation}



\section{Local Existence Theorem}\lab{sec:local-existence}

\subsection{Precise statement of the local existence}

\begin{definition}\lab{def:sphere-data}
We refer to the following as a set of canonical sphere data of radius $r_0$:
\begin{itemize}
\item A Riemannian $2$-sphere $(S,\ga)$, with a canonical uniformization form of the metric $\ga=r_0^2 e^{2u} \gacr$ in some coordinates $(\vth^A)$;
\item The radial expansion $\trth$ and the expansion in the time direction $\trt$;
\item The $1$-form $\Xi$ on $(S,\ga)$ with $\divh \Xi=0$.
\end{itemize}
\end{definition}

\begin{remark}
The sphere data defined here is analogous to what is needed for a bifurcating characteristic initial value problem. Indeed, in that setting, one needs to have the outgoing and incoming null expansions, along with the torsion $1$-form $\zeta_a:=\frac 12 \g(\D_{e_a} e_4,e_3)$, where $\{e_3,e_4\}$ is a spacetime null pair. The requirement $\divh\Xi=0$, which would correspond to a condition on $\divh\zeta$ for the characteristic problem, is legitimate in view of the ambiguity of the conformal scaling $(e_3,e_4)\mapsto (\la^{-1} e_3,\la e_4)$. 
\end{remark}

\begin{definition}\lab{eq:almost-round-sphere-data}
We say a set of canonical sphere data $(S,u,\vth^A, \trth,\trt, \Xi)$ is $O_{i}(\eps_0)$-almost round with radius $r_0$, if the data satisfies the estimate
\beaa
r_0^{-1} \| u \|_{\H_0^{i}(S)} \lesssim \eps_0,\quad \| \trth-2r_0^{-1}, \trt,\Xi \|_{\H_0^{i-1}(S)} \lesssim \eps_0. 
\eeaa
Here, the $\H_0^i$ norms are defined using the round metric in the coordinates $(\vth^A)$ as in Definition \ref{def:L2-Hs-S_r}.
\end{definition}

\begin{theorem}[Local existence]\lab{thm:local-existence-precise}
Given $s\geq 3$, consider a set of $O_{s+2}(\eps_0)$-almost round canonical sphere data $(S,u,\vth^A,\trth,\trt,\Xi)$ of radius $r_0$, as defined in \eqref{eq:almost-round-sphere-data}, with smallness parameter $\eps_0 \ll 1$.
Then, given four scalars $(\BB,\BBd, \KK,\KKd)$ on $(r_0,\infty)\times \mathbb{S}^2$,\footnote{As the argument is entirely local in $r$, it would in fact suffice to prescribe them only on a sufficiently small interval near $r_0$.}      there exists
$\eps_{loc}>0$, depending on $r_0$ and $\sup_r \|\BB,\BBd, \KK,\KKd\|_{\H^s(S_r)}$, and a set of initial data $(g,k)$ on $(r_0,r_0+\eps_{loc})\times \mathbb{S}^2$ solving \eqref{ece}, such that the metric can be written as
\begin{equation*}
g=\ah^2 dr^2+ r^2 e^{2u} \gacr_{AB} (d\vth^A +b^A dr)(d\vth^B+b^B dr),
\end{equation*}
and the data satisfies the following:
\begin{itemize}
\item The gauge conditions
\begin{itemize}
\item The $\ell=0$ conditions
\begin{equation}
\lab{eq:ell=1-condition-u-solution1}
\overline{\ao}=0,\quad \overline{\Pi}=-\frac 12 \overline{\trt}. 
\end{equation}
\item The $\ell=1$ conditions 
\begin{equation}\lab{eq:ell=1-condition-u-solution}
\int_{\mathbb{S}^2} e^{2u} x^i =0,\quad (e^{2u} \curlh b)_{\ell=1}=0;
\end{equation}
\item The general conditions on gauge scalars
\begin{equation}
\lab{eq:ell=1-condition-u-solution3}
\mu_{\ell\geq 1}=\nu=0.
\end{equation}
\end{itemize}
\item The prescribed scalar conditions
\begin{equation}\lab{eq:prescribed-scalar-conditions}
\begin{gathered}
(\divh Y-\BB)_{\ell\geq 2}=0,\quad (\curlh Y-\BBd)_{\ell\geq 2}=0, \\
\Big(\laph \big( \Pi+\frac 12 \trt \big)-\KK\Big)_{\ell\geq 2}=0,\quad \big(r^{-4} \ah N(r^4 \curlh \Xi)-\KKd\big)_{\ell\geq 2}=0.
\end{gathered}
\end{equation}
\end{itemize}
Moreover, for any positive $\eps_{loc}'<\eps_{loc}$ and a set of initial data $(g',k')$ defined on $(r_0,r_0+\eps_{loc}')\times \mathbb{S}^2$ solving \eqref{ece} and verifying    the same conditions  as $(g, k)$, we have $(g,k)|_{(r_0,r_0+\eps_{loc}')\times \mathbb{S}^2}=(g',k')$.
\end{theorem}

\begin{remark}\lab{rem:large-free-scalars}
Note that in the statement   \eqref{eq:prescribed-scalar-conditions}   of  our local  existence result, the scalars $(\BB,\BBd,\KK,\KKd)$ are allowed to be large.\footnote{The smallness assumptions on the data  are only needed  in the global result.}   In fact, in a forthcoming work  \cite{CK26-2}, we will use this important freedom  to study the spacelike short-pulse type data.
\end{remark}

\subsection{Heuristics for how to treat  the  large free scalars}\lab{sec:heuristics-large-scalars}
As already mentioned   earlier, in order to  include the regime of large free scalars,   we divide our main quantities  in the sets $\Ga$ and $ \widetilde{\Ga}$, as introduced 
 in \eqref{eq:schematic-quantities},  \eqref{eq:schematic-quantities-enlarged}   and   Remark \ref{remark;small-large}.  The small quantities in $\Ga$  have their size  inherited from the small sphere data, while the  
 size of the  large quantities in $\widetilde{\Ga}\setminus\Ga$ is induced by   the large free scalars.

There are two issues to address:
\begin{itemize} 
\item The first is to  decouple  the estimates for the small quantities  from 
the large free scalars. For the transport equations, this causes no difficulty, since all relevant estimates come with the small factor $\eps_{loc}$. Moreover, the elliptic equation for $\ao$ (due to the definition of $\mu$ in \eqref{eq:def-mu})  and the elliptic equation for $u$ (due  to \eqref{eq:laph-u-K-pert}-\eqref{eq:u-ell=1-pert})  do not involve any problematic large quantities. Therefore, in the iterative procedure of Section \ref{sec:iteration-system},   one can first control the metric coefficients $\ah\nn$ and $u\nn$, as well as the transported quantities $\thc\mkern-0.1mu\nn$, $\Kc\nn$, $\trt\nn$, and $\Xi\nn$.
In particular, we can formulate the remaining elliptic equations using $\d\nn$ rather than $\d\n$, due to the newly obtained improved control of $u\nn$.

\item The second issue is then to estimate the large quantities themselves. If $M$ denotes the largeness originating from the free scalars, applying $\d\nn$ to a large background term may produce contributions of size $u\nn\cdot M$. Thus, in the norm we use, see \eqref{eq:Psi-norms-gz-M}, the natural scale for the perturbation of large quantities is $M$ times the size of the small perturbation.

\vspace{2ex}
It remains to check that the nonlinear terms in the remaining elliptic equations can be controlled. The  point here   is that every potentially dangerous product involving a large quantity also contains one of the already improved small quantities $(\thc,\Xi,\slashed{d} (\log\ah))$. This is the structure that allows the boundedness and contraction argument to be closed in such a norm.
\end{itemize}

\subsection{The iteration system}\lab{sec:iteration-system}
We aim to prescribe scalars $(\BB,\BBd,\KK,\KKd)$, supported on $\ell\geq 2$, such that
\begin{equation*}
\begin{gathered}
\divh Y=\BB+\Bb_{\ell\leq 1},\quad \curlh Y=\BBd+\Bbd_{\ell\leq 1},\\ \laph(\Pi+\frac 12\trt)=\KK+\Kk_{\ell\leq 1},\quad r^{-4}\ah N(r^4\curlh \Xi)=\KKd+\Kkd_{\ell\leq 1}.
\end{gathered}
\end{equation*}
We define the iterate $\Psi\n$ as the collection of the following quantities
\begin{equation}
\begin{gathered}
\thc\mkern-0.1mu\n, \Kc\n, \ao\n, \thh\n, Y\n, u\n, b\n,\\
\trt\n, \kh\n, \Xi\n, \Pi\n, \\
(\Bb_{\ell\leq 1}\n,\Bbd_{\ell\leq 1}\n), (\Kk_{\ell\leq 1}\n,\Kkd_{\ell\leq 1}\n),
\end{gathered}
\end{equation}
We consider the following iteration system, with $N\n:=(\ah\n)^{-1} (\pa_r - b\n\cdot \nabz)$, where $b\n$ is constructed as a $1$-form:\footnote{In other words, we are seeking the limit $b\i$ as $b^{\flat,0}$ defined in \eqref{eq:def-b-flat-0}.} 
\begin{align}
\lab{eq:R-transport-kac2}
N\n \thc\mkern-0.1mu\nn &= \mu_{\ell=0}\n -2r^{-1}\thc\mkern-0.1mu\n-2 r^{-2}\ao\n +\widetilde\Ga_1\n\cdot \widetilde\Ga_1\n, \\
\lab{eq:R-transport-P2}
N\n \Kc\nn_{\ell\neq 1} &= r^{-1}\mu\n_{\ell=0} - \left(\BB+\Bb_{\ell\leq 1}\n\right)_{\ell\neq 1}-3r^{-1}\Kc\n_{\ell\neq 1} -2 r^{-3} \ao\n_{\ell\neq 1} \\
&\quad \nonumber +\left(\Ga_1\n\cdot \Ga_2\n+2\, \slashed{d}(\log\ah\n)\cdot Y\n\right)_{\ell\neq 1}+\RR_{tran}\n (\Kc\n) \\
&\quad \notag -\left(\thh\n\cdot \left(\nabh\n\hot\, \slashed{d}(\log\ah\n)-\slashed{d}(\log\ah\n)\hot\, \slashed{d}(\log\ah\n)\right)\right)_{\ell\neq 1},\\
\lab{eq:iterate-laph-ao}
\laph\n \ao\nn &= \Kc\nn-\overline{\Kc\nn}\n- r^{-1} \Big(\thc\mkern-0.1mu\nn-\overline{\thc\mkern-0.1mu\nn}\n\Big) \\
&\quad \nonumber -\laph\n(\Ga_0\n\cdot \Ga_0\n)+\Ga_1\n\cdot\Ga_1\n-\overline{\Ga_1\n\cdot\Ga_1\n}\n, \\
\lab{eq:spherical-mean-ao2}
\overline{\ao\nn}\n &= 0,\\
 \lab{eq:iterate-laph-u-K}
  -\laph_{\gacr} u\nn-2 u\nn &= r^2 \Kc\nn+r^2\Kc\n (e^{2u\n}-1)+(e^{2u\n}-1-2u\n),\\
  \lab{eq:iterate-ell=1-of-u}
 u\nn_{\ell=1}&=-\frac 12 (e^{2u\n}-1-2u\n)_{\ell=1},
 \end{align}
 \begin{align}
 \lab{eq:unconditional-Codazzi-1-2}
\d_1\nn \d_2\nn \thh\nn &= \Big(\frac 12 \laph\n\thc\mkern-0.1mu\nn,0\Big)-(\BB+\Bb_{\ell\leq 1}\nn,\BBd+\Bbd_{\ell\leq 1}\nn),\\
 \lab{eq:iterate-nab-hot-b}
 \nabz \hot b\nn & = -2e^{-2u\nn} \ah\nn \thh\nn, \\
 \lab{eq:iterate-pa-r-u-ell=1}
\left(\divz  b\nn \right)_{\ell=1}&= -\left(\divz ((e^{2u\nn}-1)b\n)\right)_{\ell=1}\\
&\notag \quad +\left(e^{2u\nn} (2r^{-1}\ao\nn+\thc\mkern-0.1mu\nn+\ao\n\cdot \thc\mkern-0.1mu\n)\right)_{\ell=1},\\
  \lab{eq:iterate-ell=1-of-curl-b}
 \left(\curlz b\nn\right)_{\ell=1} &=- \left(\curlz ((e^{2u\nn}-1)b\n)\right)_{\ell=1}, \\
\lab{eq:iterate-Y}
\d_1\nn Y\nn &=  (\BB, \BBd)+(\Bb_{\ell\leq 1}\nn,\Bbd_{\ell\leq 1}\nn)  -\overline{\displaystyle (\BB, \BBd)+(\Bb_{\ell\leq 1}\nn,\Bbd_{\ell\leq 1}\nn)}\nn,
\end{align}
\begin{align}
\lab{eq:N-trt-linearized-iterate}
N\n\trt\nn &= 2 r^{-1} \Pi\n  - r^{-1} \trt\n + \widetilde\Ga_1\n\cdot \widetilde\Ga_1\n ,\\
\lab{eq:nu-iterate}
\divh\n \Xi\nn &=0, \\
\lab{eq:curl-Xi-iterate-1}
\ah\n N\n c\nn  &=-4r^{-1} c\nn+ \KKd-\Kkd_{\ell\leq 1}\n,\\
\lab{eq:curl-Xi-iterate-2}
\curlh\n \Xi\nn &= c\nn-\overline{c\nn}\n,\\
\lab{eq:d1-d2-Thh-iterate}
\d_1\nn \d_2\nn (\ah\n \Thh\nn) &= (\KK, -\KKd)+(\Kk_{\ell\leq 1}\nn,\Kkd_{\ell\leq 1}\nn)+\Ga_1\n\cdot \Ga_2\n \\
&\notag\quad +\nabh\n\left( (\thh\n,\Thh\n,\Pi\n)\cdot (\slashed{d} (\log\ah\nn),\Xi\nn)\right)\\
&\notag\quad +Y\n\cdot \Xi\nn,\\
\lab{eq:iteration-Pi}
\laph\nn \left(\Pi\nn+ \frac 12 \trt\nn \right) &= \KK + \Kk_{\ell\leq 1}\nn-\overline{\KK+ \Kk_{\ell\leq 1}\nn}\nn,\\
\lab{eq:average-Pi}
\overline{\Pi\nn}\nn&= -\frac 12 \overline{\trt\nn}\nn. 
\end{align}
Note that for the transport equations \eqref{eq:R-transport-kac2}-\eqref{eq:R-transport-P2}, \eqref{eq:N-trt-linearized-iterate} and \eqref{eq:curl-Xi-iterate-1}, their initial data at $r=r_0$ is given by the canonical sphere data---for $c\n$, we let $c\n|_{r=r_0}=(\curlh \Xi)|_{r=r_0}$. For \eqref{eq:R-transport-P2}, we also note that $\RR_{tran}\n$ only depends on $(a\n,b\n,u\n)$, according to the proof of Lemma \ref{lem:RR-transport}.

Here, the horizontal operators label with $\n$ are defined with respect to the metric $\ga\n$ through the expression
\bea\lab{eq:iterate-ga}
\ga\n:= r^2 \left(e^{2u\n}\right) \gacr.
\eea
The $3$-dim metric at the $n$-th step is defined as
\bea\lab{eq:iterate-g}
g\n:= (\ah\n)^2 dr^2 +\ga\n_{AB} \left(d\vth^A +(b\n)^A dr\right) \left(d\vth^B+(b\n)^B dr\right).
\eea

\subsection{Boundedness estimates}

We deal with the case when
\beaa
M:=\sup_r r^{-1} \| \BB,\BBd,\KK,\KKd \|_{\H^s(S_r)}
\eeaa can be a large number.
As a result, $b$, $\thh$, $\Thh$, $\Pi$ and $Y$ are not necessarily small. We therefore consider the following checked quantities
\bea\lab{eq:def-widecheck-large-quantities}
\widecheck b:=b-b_0,\quad \widecheck{\thh}:= \thh-\thh_0,\quad \widecheck \Thh:= \Thh-\Thh_0,\quad \widecheck \Pi:=\Pi-\Pi_0,\quad Y:= Y-Y_0,
\eea
where $b_0$, $\thh_0$, $\Thh_0$, $\Pi_0$, and $Y_0$ are determined through the following equations:
\begin{equation}\lab{eq:def-zero-quantities}
\begin{gathered}
\d_1\0 \d_2\0 \thh_0=(\BB,\BBd),\quad \d_1\0\d_2\0 \Thh_0= (\KK,-\dual \KK),\quad \d_1\0 Y_0 = (\BB,\BBd),\\
 \lapz \Pi_0= \KK,\quad (\Pi_0)_{\ell=0}=0, \\
 \nabz\hot b_0=\thh_0,\quad (\divz b_0)_{\ell=1}=(\curlz b_0)_{\ell=1}=0.
\end{gathered}
\end{equation}
\begin{remark}\lab{rem:M-eps-0}
In the case when $M\lesssim \eps_0$,   we can proceed without subtracting $b_0$, $\thh_0$, $\Thh_0$, $\Pi_0$, and $Y_0$ in \eqref{eq:def-widecheck-large-quantities}.
\end{remark}

Recall the definition of $\Ga_i$ and $\widetilde \Ga_i$, $i=0,1,2$,  in \eqref{eq:schematic-quantities} and \eqref{eq:schematic-quantities-enlarged}.
We furthermore denote 
\begin{equation*}
\widecheck{\widetilde \Ga}_0= \{\ao,\widecheck b, u\},\quad \widecheck{\widetilde \Ga}_1=\{r^{-1} \widecheck{\widetilde \Ga}_0,\thc, \slashed{d} (\log\ah),\Xi,\trt,\widecheck \thh, \widecheck\Thh, \widecheck\Pi\},\quad \widecheck{\widetilde \Ga}_2= \{r^{-1}\widecheck{\widetilde \Ga}_1,\nabh\widecheck{\widetilde \Ga}_1,\Kc, \widecheck Y\}.
\end{equation*}
Define
\begin{equation}\lab{eq:Psi-norms-gz-M}
\bsplit
\| \Psi \|_{s,0,  I,M} &:=  \sup_{r\in I}\, r^{-1} \| \ao, \widecheck b, u \|_{\H^{s+2}_0(S_r)}+ \| \thc, \trt, \Xi, \langle M\rangle^{-1} \widecheck \thh, \langle M\rangle^{-1} \widecheck\Thh, \langle M\rangle^{-1} \widecheck\Pi \|_{\H^{s+1}_0(S_r)} \\
&\quad + r \| \Kc, \langle M\rangle^{-1} \widecheck Y \|_{\H_0^s(S_r)}+ \langle M\rangle^{-1} r^3 |(\Bb_{\ell\leq 1},\Bbd_{\ell\leq 1},\Kk_{\ell\leq 1},\Kkd_{\ell\leq 1})|. 
\end{split}
\end{equation}

We have the following boundedness result. 
\begin{proposition}\lab{prop:boundedness}
For the $O_{s+2}(\eps_0)$-almost round canonical sphere data $(S,u,\vth^A,\trth,\trt,\Xi)$ of radius $r_0$ and $(\BB,\BBd,\KK,\KKd)$ as considered in the statement of Theorem \ref{thm:local-existence-precise}, there exists $\eps_{loc}>0$, such that if 
  $\Psi\n$ satisfies,  for some $n$,
\begin{align*}
\| \Psi\n \|_{s,0,[r_0,r_0+\eps_{loc}],M} \leq \eps_0^{\frac 23},
\end{align*} 
then the system \eqref{eq:R-transport-kac2}-\eqref{eq:average-Pi}, together with the initial conditions on $\{r=r_0\}$ determined by the sphere data, admits a unique solution $\Psi\nn$ satisfying
\begin{align*}
\| \Psi\nn \|_{s,0,[r_0,r_0+\eps_{loc}],M} \leq \eps_0^{\frac 23}.
\end{align*}
\end{proposition}

\begin{proof}
According to the assumptions, the relevant sizes are 
\[
\Ga\n \sim \eps_0^\frac 23,\quad \widecheck{\widetilde \Ga}\mkern-0.1mu\n \sim \eps_0^\frac 23 M, \quad \widetilde \Ga\n \sim M.\] 
We proceed as follows:\footnote{
In the following, we repeatedly use standard product estimates, obtained by combining Hölder's inequality with Sobolev embedding to control one factor in $L^\infty$ with fewer derivatives.
}
\begin{enumerate}
\item\lab{step-boundedness-thc}
 Integrating \eqref{eq:R-transport-kac2}-\eqref{eq:R-transport-P2} and applying \eqref{eq:transport-estimate-gz} in Lemma \ref{lem:integration-lemma-1}, we obtain
\beaa
&& r^{-1} \| \thc\mkern-0.1mu\nn\|_{\H_0^{s+1}(S_r)} \lesssim r_0^{-1} \| \thc\mkern-0.1mu\nn\|_{\H_0^{s+1}(S_{r_0})} +\eps_{loc}\, M^2 \lesssim r_0^{-1} \eps_0+\eps_{loc}\, M^2 ,\\
&& r^{-1} \| \Kc\nn_{\ell\neq 1} \|_{\H_0^s(S_r)} \lesssim r_0^{-1} \| \Kc\nn_{\ell\neq 1} \|_{\H_0^s(S_{r_0})}+\eps_{loc}\, M^2 \lesssim r_0^{-2} \eps_0+\eps_{loc}\, M^2,
\eeaa
where we used that the sphere data at $r=r_0$ is $O(\eps_0)$.
Note that $r r_0^{-1}=r_0^{-1} (r_0+O(\eps_{loc}))=1+r_0^{-1} O(\eps_{loc})$, and similarly for $r^2 r_0^{-2}$.
Therefore, we deduce
\bea\lab{eq:improved-estimate-thc}
 \| \thc\mkern-0.1mu\nn\|_{\H_0^{s+1}(S_r)}, r \|\Kc_{\ell\neq 1}\nn \|_{\H_0^s(S_r)} \lesssim \eps_0+\eps_{loc}\, M^2 .
\eea
\item\lab{step-boundedness-u} We determine $\Kc_{\ell=1}$ by projecting \eqref{eq:iterate-laph-u-K} to $\ell=1$. This is standard since the Laplacian in \eqref{eq:iterate-laph-u-K} is with respect to the round metric:
\beaa
r^2 | \Kc_{\ell=1}\nn | \lesssim | r^2 (\Kc\n \cdot u\n)_{\ell=1} |+|(u\n\cdot u\n)_{\ell=1}| \lesssim  \eps_0^\frac 43 .
\eeaa
Combining this with the condition \eqref{eq:iterate-ell=1-of-u}, which reads schematically $u\nn_{\ell=1}=u\n\cdot u\n$, we also obtain
\bea\lab{eq:improved-estimate-u}
r^{-1} \| u\nn \|_{\H_0^{s+2}(S_r)} \lesssim \eps_0^\frac 43 + r \|\Kc\nn \|_{\H_0^s(S_r)} \lesssim \eps_0+\eps_{loc}\, M^2 .
\eea
Note that no quantities in $\widetilde \Ga \backslash \Ga$ are involved in the equations \eqref{eq:iterate-laph-u-K}-\eqref{eq:iterate-ell=1-of-u}.
\item\lab{step-boundedness-ah}
 We then determine $\ah\nn$ from \eqref{eq:iterate-laph-ao} and \eqref{eq:spherical-mean-ao2}, using the fact that $\thc\mkern-0.1mu\nn$ and $\Kc\nn$ have been obtained. 
 This yields
\bea\lab{eq:improved-bound-ah}
r^{-1}\| \ao\nn \|_{\H^{s+2}_0(S_r)} \lesssim \eps_0+\eps_{loc} M^2 .
\eea
\item 
The checked quantity $\widecheck{\thh}\mkern-0.1mu\nn$ satisfies the equation, according to \eqref{eq:unconditional-Codazzi-1-2} and \eqref{eq:def-zero-quantities},
\begin{equation}\lab{eq:d1d2-widecheck-thh}
\bsplit
\d_1\nn \d_2\nn \widecheck{\thh}\mkern-0.1mu\nn &= (\d_1\0\d_2\0-\d_1\nn\d_2\nn) \thh_0\\
&\quad +\Big(\frac 12 \laph\n\thc\mkern-0.1mu\nn,0\Big)-(\Bb_{\ell\leq 1}\nn, \Bbd_{\ell\leq 1}\nn). 
\end{split}
\end{equation}
According to Lemma \ref{lem:comparison}, there exists a unique pair $(\Bb_{\ell\leq 1}\nn,\Bbd_{\ell\leq 1}\nn)$ satisfying
\begin{equation}\lab{eq:estimate-Bd-Bbd-thc-ell=1}
\bsplit
& \quad \; r^3 | (\Bb_{\ell\leq 1}\nn,\Bbd_{\ell\leq 1}\nn)+ r^{-2} (\thc\mkern-0.1mu\nn_{\ell=1},0) |\\
&\lesssim  r^2 \|(\d_1\0\d_2\0-\d_1\nn\d_2\nn) \thh_0\|_{L_0^2(S_r)}+r|u\n \cdot \thc\mkern-0.1mu\n | \\
&\lesssim  (\eps_0+M^2\eps_{loc}) M+\eps_0^\frac 43,
\end{split}
\end{equation}
for which \eqref{eq:d1d2-widecheck-thh} admits a unique solution $\widecheck \thh\mkern-0.1mu\nn$. Note that we used the improved bound \eqref{eq:improved-estimate-u} for $u\nn$ to estimate the difference $\d_1\0\d_2\0-\d_1\nn\d_2\nn$. Combining \eqref{eq:estimate-Bd-Bbd-thc-ell=1} with \eqref{eq:improved-estimate-thc}, we obtain the same bound for $(\Bb_{\ell\leq 1}\nn,\Bbd_{\ell\leq 1}\nn)$. Moreover, by the Hodge estimate in Lemma \ref{lem:comparison}, we have
\beaa
\| \widecheck\thh\mkern-0.1mu\nn \|_{\H_0^{s+1}(S_r)} &\lesssim & \| \thc\mkern-0.1mu\nn \|_{\H_0^{s+1}(S_r)}+r^2 \|(\d_1\0\d_2\0-\d_1\nn\d_2\nn) \thh_0\|_{\H_0^{s-1}(S_r)} \\
&\lesssim & \eps_0+\eps_{loc} M^2+(\eps_0+\eps_{loc} M^2)M \lesssim (\eps_0+\eps_{loc} M^2)M,
\eeaa
where we again used the improved bound for $\thc\mkern-0.1mu\nn$ and $u\nn$ in \eqref{eq:improved-estimate-thc} and \eqref{eq:improved-estimate-u}. 
\item In view of \eqref{eq:iterate-pa-r-u-ell=1}, \eqref{eq:iterate-ell=1-of-curl-b} and the fact that $u\nn$, $\ah\nn$, and $\thc\mkern-0.1mu\nn$ have already been determined, we obtain the conditions for $(\divz b\nn)_{\ell=1}$ and $(\curlz b\nn)_{\ell=1}$. 
Thanks to \eqref{eq:def-zero-quantities}, the same conditions hold for $(\divz \widecheck b\nn)_{\ell=1}$ and $(\curlz \widecheck b\nn)_{\ell=1}$. We also have, using \eqref{eq:iterate-nab-hot-b} and \eqref{eq:def-zero-quantities},
\beaa
\nabz\hot \widecheck b\nn=(\ah\nn) \thh\nn-2(e^{2u\nn} \ah\nn-1)\cdot \thh_0.
\eeaa
We can then solve \eqref{eq:iterate-nab-hot-b} for $\widecheck b\nn$ using the standard Hodge theory on the round sphere, which yields, using the improved bounds for $\ao\nn$ and $u\nn$,
\beaa
r^{-1} \| \widecheck b\nn \|_{\H_0^{s+2}(S_r)} \lesssim \| \widecheck \thh\mkern-0.1mu\nn \|_{\H_0^{s+1}(S_r)} + (\eps_0+ M^2 \eps_{loc}) M\lesssim (\eps_0+ M^2 \eps_{loc}) M.
\eeaa
\item Using  \eqref{eq:iterate-Y} and \eqref{eq:def-zero-quantities}, the equation for $\widecheck Y\nn$ reads
\begin{align*}
\d_1\nn \widecheck Y\nn &= (\d_1\0-\d_1\nn) Y_0+( \Bb_{\ell\leq 1}\nn, \Bbd_{\ell\leq 1}\nn)\\
&\quad -\overline{\displaystyle(\BB,\BBd)+(\Bb_{\ell\leq 1}\nn, \Bbd_{\ell\leq 1}\nn)}\nn.
\end{align*}
In view of the improved estimate for $u\nn$ and $(\Bb_{\ell\leq 1}\nn, \Bbd_{\ell\leq 1}\nn)$, and the fact that $(\BB,\BBd)$ have zero average with respect to $\gz$, all terms on the right hand side are $\lesssim (\eps_0+M^2 \eps_{loc}) M$. Using the Hodge estimate in Lemma \ref{lem:comparison}, we deduce
\beaa
r \|\widecheck Y\nn \|_{\H_0^s(S_r)} &\lesssim & r^2 \| (\d_1\0-\d_1\nn) Y_0 \|_{\H_0^{s-1}(S_r)}+ (\eps_0+M^2 \eps_{loc}) M \\
&\lesssim & (\eps_0+M^2 \eps_{loc}) M.
\eeaa
\item Integrating \eqref{eq:N-trt-linearized-iterate} and \eqref{eq:curl-Xi-iterate-1} respectively in a similar way as in Step \ref{step-boundedness-thc}, 
we obtain 
 \bea\lab{eq:improved-bound-trTh}
 \| \trt\nn \|_{\H_0^{s+1}(S_r)}, r \| c\nn \|_{\H_0^s(S_r)} \lesssim \eps_0+ \eps_{loc}\, M^2. 
 \eea
\item\lab{step-boundedness-Xi}
In view of \eqref{eq:curl-Xi-iterate-2} and \eqref{eq:nu-iterate}, we obtain $\Xi\nn$ and, according to the Hodge estimates in Lemma \ref{lem:comparison}, we have the estimate
\bea\lab{eq:improved-bound-Xi}
\|\Xi\nn \|_{\H_0^{s+1}(S_r)} \lesssim \eps_0+\eps_{loc}\, M^2. 
\eea
\item 
Using \eqref{eq:d1-d2-Thh-iterate} and \eqref{eq:def-zero-quantities}, we write
\begin{equation}\lab{eq:Codazzi-widecheck-Thh}
\bsplit
&\quad  \d_1\nn\d_2\nn (\ah\n \Thh\nn-\Thh_0) \\
 &= \nabh\n\left( (\thh\n,\Thh\n,\Pi\n)\cdot (\slashed{d} (\log\ah\nn),\Xi\nn)\right)\\
&\quad +Y\n\cdot \Xi\nn+\Ga_1\n\cdot \Ga_2\n \\
&\quad + (\d_1\0\d_2\0-\d_1\nn\d_2\nn)\Thh_0 +(\Kk_{\ell\leq 1}\nn, \Kkd_{\ell\leq 1}\nn).
\end{split}
\end{equation}
Note that $u\nn$, $\ah\nn$ and $\Xi\nn$ have been obtained in Steps \ref{step-boundedness-u}, \ref{step-boundedness-ah} and \ref{step-boundedness-Xi}, and we can use their improved bounds \eqref{eq:improved-estimate-u}, \eqref{eq:improved-bound-ah} and \eqref{eq:improved-bound-Xi}.
To solve \eqref{eq:Codazzi-widecheck-Thh} for $\ah\n \Thh\nn-\Thh_0$, we again apply Lemma \ref{lem:comparison} to obtain the unique pair $(\Kk_{\ell\leq 1}\nn,\Kkd_{\ell\leq 1}\nn)$ satisfying the estimate
\beaa
&& r^3 |(\Kk_{\ell\leq 1}\nn,\Kkd_{\ell\leq 1}\nn)| \\
&\lesssim & M \| \slashed{d} (\log\ah\nn),\Xi\nn \|_{\H^1_0(S_r)}+|(r\nabz)^{\leq 1}u\nn|\cdot |(r\nabz)^{\leq 2}\Thh_0|+\eps_0^\frac 43 \\
&\lesssim & (\eps_0+M^2 \eps_{loc}) M+\eps_0^\frac 43.
\eeaa
The corresponding unique solution $\ah\n \Thh\nn-\Thh_0$ to \eqref{eq:Codazzi-widecheck-Thh} then satisfies, by Lemma \ref{lem:comparison},
\beaa
&&\| \ah\n \Thh\nn-\Thh_0 \|_{\H_0^{s+1}(S_r)} \\
& \lesssim & M \| \slashed{d} (\log\ah\nn),\Xi\nn \|_{\H^s_0(S_r)}+ r^{-1}\|u\nn\|_{\H^s_0(S_r)} \| \Thh_0 \|_{\H_0^{s+1}(S_r)}+\eps_0^\frac 43\\
&\lesssim & (\eps_0+M^2 \eps_{loc}) M+\eps_0^\frac 43.
\eeaa
Then, using the assumed bound for $\ao\n=\ah\n-1$, we obtain the same estimate with $\ah\n \Thh\nn-\Thh_0$ replaced by $\widecheck\Thh\mkern-0.1mu\nn$.

\item 
We now have the following elliptic equation, according to \eqref{eq:iteration-Pi}-\eqref{eq:average-Pi} together with \eqref{eq:def-zero-quantities},
 \beaa
 \laph\nn (\widecheck \Pi\nn+\frac 12\trt\nn) &=& (\laph\nn-\lapz)\Pi_0 ,\\
 \quad \overline{\displaystyle \widecheck \Pi\nn+\frac 12\trt\nn}\nn &=& 0.
 \eeaa
Similar to the previous step, using the improved estimate for $u\nn$, we deduce
 \beaa
\| \widecheck \Pi\nn+\frac 12 \trt\nn \|_{\H_0^{s+1}(S_r)} \lesssim r^{-1} M \| u\nn \|_{\H_0^s(S_r)} \lesssim (\eps_0+M^2 \eps_{loc}) M,
 \eeaa 
 which then also yields $\|\widecheck \Pi\nn \|_{\H_0^{s+1}(S_r)} \lesssim (\eps_0+M^2 \eps_{loc}) M$ in view of the estimate for $\trt\nn$ in \eqref{eq:improved-bound-trTh}.
\end{enumerate}
Therefore, according to the definition of the norm in \eqref{eq:Psi-norms-gz-M}, we have deduced \[\| \Psi\nn \|_{s,0,[r_0,r_0+\eps_{loc}],M} \lesssim \eps_0+M^2 \eps_{loc} \ll \eps_0^{\frac 23},\] provided $\eps_{loc}>0$ is sufficiently small. This concludes the proof.
\end{proof}

\subsection{Contraction estimates}
In this section, we denote $\de\Psi\nn:=\Psi\nn-\Psi\n$, and $F\n[\phi]:= N\n \phi\nn$, i.e., the right-hand side of the transport equation of a quantity $\phi\nn$ in the iteration system \eqref{eq:R-transport-kac2}-\eqref{eq:iterate-pa-r-u-ell=1}. Note that $F\n$ indeed only contains quantities labeled with $\n$.  
\begin{lemma}\lab{lemma:difference-contraction}
For an equation $N\n \phi\nn=F\n$, we have
\begin{equation}\lab{eq:N-contraction}
\bsplit
N\nn (\de \phi\nnn) &= \de F\nn - ((\ah\nn)^{-1} - (\ah\n)^{-1} ) \ah\n F\n \\
&\quad +(\ah\nn)^{-1} (b\nn - b\n) \cdot \nabz \phi\nn.
\end{split}
\end{equation}
\end{lemma}
\begin{proof}
We have
\beaa
&& \de F\nn= N\nn \phi\nnn -N\n \phi\nn \\
&=&  (\ah\nn)^{-1} (\pa_r -  b\nn\cdot \nabz)\phi\nnn-(\ah\n)^{-1} (\pa_r - b\n\cdot \nabz)\phi\nn \\
&=& (\ah\nn)^{-1} \pa_r (\phi\nnn-\phi\nn)- (\ah\nn)^{-1} b\nn\cdot\nabz (\phi\nnn-\phi\nn) \\
&& +((\ah\nn)^{-1} - (\ah\n)^{-1} ) \pa_r \phi\nn- ((\ah\nn)^{-1} b\nn-(\ah\n)^{-1} b\n)\cdot \nabz \phi\nn \\
&=& N\nn (\de \phi\nnn) +((\ah\nn)^{-1} - (\ah\n)^{-1} ) \pa_r \phi\nn \\
&& -((\ah\nn)^{-1} b\nn-(\ah\n)^{-1} b\n)\cdot \nabz \phi\nn.
\eeaa
Now, using $\pa_r \phi\nn= \ah\n N\n \phi\nn+ b\n \cdot \nabz \phi\nn$, we obtain
\beaa
N\nn (\de \phi\nnn) &=& \de F\nn - ((\ah\nn)^{-1} - (\ah\n)^{-1} ) \ah\n N\n \phi\nn \\
&& - ((\ah\nn)^{-1} - (\ah\n)^{-1} ) b\n \cdot \nabz \phi\nn \\
&&+((\ah\nn)^{-1}b\nn-(\ah\n)^{-1} b\n)\cdot \nabz \phi\nn \\
&=& \de F\nn - ((\ah\nn)^{-1} - (\ah\n)^{-1} ) \ah\n N\n \phi\nn \\
&& +(\ah\nn)^{-1} (b\nn - b\n) \cdot \nabz \phi\nn.
\eeaa
Using again $N\n \phi\nn=F\n$, we obtain the desired identity.
\end{proof}
We also have, by straightforward verification, for any sequence of horizontal operators $L\n$,
\begin{equation}\lab{eq:general-formula-L-difference}
\bsplit
&\quad L\nn[\phi\nnn]-L\n[\phi\nn] \\
&= L\nn[{\de\phi\nnn}] + (L\nn-L\n)[\phi\nn],
\end{split}
\end{equation}
\begin{equation}\lab{eq:general-formula-L-difference-2}
\bsplit
&\quad L\nnn[\phi\nnn]-L\nn[\phi\nn] \\
&= L\nnn[{\de\phi\nnn}] + (L\nnn-L\nn)[\phi\nn].
\end{split}
\end{equation}
This applies for $L\n=\laph\n$, $\d_1\n$, $\d_2\n$, as well as the spherical mean operator $\phi\mapsto \overline{\phi}\n$. 

\begin{proposition}
There exists a universal constant $C>0$ such that 
\beaa
\| \de\Psi\nnn \|_{s-1} \leq C (\eps_{loc}M^2 +\eps_0^\frac 23) \| \de\Psi\nn \|_{s-1}.
\eeaa
Here, we denote for simplicity $\| \de\Psi \|_{s-1}:= \| \de\Psi \|_{s-1,0,[r_0,r_0+\eps_{loc}],M}$ defined in \eqref{eq:Psi-norms-gz-M}.
\end{proposition}
\begin{proof}
We proceed as follows, where we use the identities \eqref{eq:N-contraction}, \eqref{eq:general-formula-L-difference}, and \eqref{eq:general-formula-L-difference-2} whenever needed:
\begin{enumerate}
\item\lab{item:step-contraction-thc}
 The equation of $\de\thc\mkern-0.1mu\nnn$ reads, schematically,
\beaa
N\nn (\de\thc\mkern-0.1mu\nnn)=\nabz \thc\mkern-0.1mu\nn\cdot \de b\nn+ \de\ah\nn\cdot F\n[\thc]+\de F\nn[\thc],
\eeaa
with vanishing initial value. Note that in view of \eqref{eq:R-transport-kac2}, the term $\de F\nn[\thc]$ consists, at worst, of terms of the form $\widetilde\Ga\nn\cdot \de \widecheck{\widetilde\Ga}\mkern-0.1mu\nn$.
Therefore, using the boundedness estimate,
we obtain\footnote{We use $\|\de \widecheck{\widetilde\Ga}\mkern-0.1mu\nn\|_{s-1}$ to denote the portion of $\|\de\Psi\nn\|_{s-1}$ associated with the quantities in $\widecheck{\widetilde\Ga}\mkern-0.1mu\nn$.}
\begin{equation}\lab{eq:de-thc-contraction}
\| \de\thc\mkern-0.1mu\nnn \|_{\H_0^s(S_r)} \lesssim \eps_{loc} M \|\de \widecheck{\widetilde\Ga}\mkern-0.1mu\nn\|_{s-1} \lesssim  \eps_{loc} M^2 \| \de\Psi\nn \|_{s-1}.
\end{equation}
\item The equation of $\de\Kc\nnn_{\ell\neq 1}$ reads, schematically,
\beaa
N\nn (\de\Kc\nnn_{\ell\neq 1})=\nabz \Kc_{\ell\neq 1}\nn\cdot \de b\nn+\de \ah\nn \cdot F\n[\Kc_{\ell\neq 1}]+ \de F\nn[\Kc_{\ell\neq 1}],
\eeaa
with vanishing initial value, and we obtain, similarly to Step \ref{item:step-contraction-thc},
\bea\lab{eq:estimate-de-Kc-ell-neq-1}
\| \de\Kc\nnn_{\ell\neq 1} \|_{\H_0^{s-1}(S_r)} \lesssim \eps_{loc} M^2 \| \de\Psi\nn \|_{s-1}.
\eea
\item The equation of $\de u\nnn$ reads
\bea\lab{eq:elliptic-equation-de-u}
-\laph_{\gacr} \de u\nnn -2\de u\nnn=r^2 \de\Kc\nnn+r^2 (u\nn\cdot \Kc\nn - u\n\cdot \Kc\n),
\eea
where $u\nn\cdot \Kc\nn$ and  $u\n\cdot \Kc\n$ have the same schematic form.
Projecting \eqref{eq:elliptic-equation-de-u} to $\ell=1$, we deduce the estimate
\bea\lab{eq:estimate-de-Kc-ell=1}
r^2 |  \de\Kc\nnn_{\ell=1} | \lesssim \eps_0^\frac 23 \| \de\Psi\nn \|_s.
\eea
In view of $u\nnn_{\ell=1}=u\nn \cdot u\nn$ and $u\nn_{\ell=1}=u\n\cdot u\n$ with the same schematic form, we have the control of $\de u\nnn_{\ell=1}$ using $\de\Psi\nn$. Then, applying standard elliptic estimates to \eqref{eq:elliptic-equation-de-u}, we deduce, using the bounds \eqref{eq:estimate-de-Kc-ell-neq-1} and \eqref{eq:estimate-de-Kc-ell=1} we obtained,
\bea\lab{eq:estimate-de-u}
\| \de u\nnn \|_{\H_0^{s+1}(S_r)} \lesssim (\eps_0^\frac 23+\eps_{loc} M^2) \| \de\Psi\nn \|_{s-1}.
\eea
\item We now turn to the equation for $\de\thh\nn$, derived from \eqref{eq:unconditional-Codazzi-1-2}, which is symmetric traceless since $\ga\nnn$ and $\ga\nn$ are conformally equivalent,
\begin{equation}\lab{eq:Codazz-de-thh}
\bsplit
\d_1\nnn\d_2\nnn \de\thh\nnn &= \frac 12 (\laph\nn \de\thc\mkern-0.1mu\nnn,0)-\de(\Bb_{\ell\leq 1}\nnn,\Bbd_{\ell\leq 1}\nnn) \\
& \quad -(\d_1\nnn \d_2\nnn-\d_1\nn \d_2\nn) \thh\nn \\
&\quad +\frac 12((\laph\nn-\laph\n)\thc\mkern-0.1mu\nn,0).
\end{split}
\end{equation}
In view of Lemma \ref{lem:comparison}, the fact that \eqref{eq:Codazz-de-thh} is solvable as a Hodge system for $\de\thh\nnn$ yields the bound 
\beaa
&& r^3 | \de(\Bb_{\ell\leq 1}\nnn,\Bbd_{\ell\leq 1}\nnn) | \\
&\lesssim &  \| \de\thc\mkern-0.1mu\nnn \|_{L_0^2(S_r)}+ r^{-1} M \| (r\nabh)^{\leq 1} \de u\nnn\|_{L_0^2(S_r)}+\eps_0^\frac 23 \| \de\Psi\nn\|_{s-1} \\
 & \lesssim & M (\eps_{loc}M^2 +\eps_0^\frac 23) \| \de\Psi\nn\|_{s-1},
\eeaa
where we use \eqref{eq:de-thc-contraction} and \eqref{eq:estimate-de-u} for the last inequality.
Moreover, using the Hodge estimate in Lemma \ref{lem:comparison}, we obtain in a similar way that
\beaa
&& \| \de\thh\nnn \|_{\H_0^{s}(S_r)} \\
& \lesssim & \| \de\thc\mkern-0.1mu\nnn\|_{\H_0^{s}(S_r)}+r^{-1} M \| \de u\nnn\|_{\H_0^{s-1}(S_r)}+\eps_0^\frac 23 \| \de\Psi\nn\|_{s-1} \\
&\lesssim & M (\eps_{loc}M^2+\eps_0^\frac 23) \| \de\Psi\nn\|_{s-1}.
\eeaa
\item 
In view of \eqref{eq:iterate-nab-hot-b}-\eqref{eq:iterate-ell=1-of-curl-b}, we can derive the equation for $\nabz\hot (\de b\nnn)$ along with the conditions of $(\divz \de b\nnn)_{\ell=1}$ and $(\curlz \de b\nnn)_{\ell=1}$, in terms of $\de u\nnn$, $\de \ah\nnn$, and $\de\thh\nnn$, which have all been estimated in previous steps. Therefore, we can proceed, similarly to the estimate for $\de\thh\nnn$, to obtain
\beaa
r^{-1} \| \de b\nnn \|_{\H_0^{s+1}(S_r)} \lesssim M (\eps_{loc}M^2 +\eps_0^\frac 23) \| \de \Psi\nn\|_{s-1}.
\eeaa 
\item Since $\de\trt\nnn$ and $\de c\nnn$, in view of \eqref{eq:N-trt-linearized-iterate} and \eqref{eq:curl-Xi-iterate-1}, satisfy transport equations, their estimates are almost identical to the one for $\de\thc\mkern-0.1mu\nnn$, and we deduce
\bea\lab{eq:estimate-de-trt-de-c}
\| \de \trt\nnn \|_{\H_0^s(S_r)}, r \| \de c\nnn \|_{\H_0^{s-1}(S_r)} \lesssim \eps_{loc} M^2 \| \de\Psi\nn \|_{s-1}.
\eea
Now, in view of the equations for $\curlh\nn \de \Xi\nnn$ and $\divh\nn \de\Xi\nnn$, derived from \eqref{eq:nu-iterate} and \eqref{eq:curl-Xi-iterate-2},
\begin{align*}
\curlh\nn \de \Xi \nnn &= -(\curlh\nn-\curlh\n)\Xi\nnn+ \de c\nnn-\overline{\de c\nnn}\nn \\
&\quad +(\overline{c\nn}\nn-\overline{c\nn}\n),\\
 \divh\nn \de\Xi\nnn &= -(\divh\nn-\divh\n)\Xi\nnn,
\end{align*}
we deduce, using \eqref{eq:estimate-de-trt-de-c}, \eqref{eq:estimate-de-u} and the boundedness estimates,
\beaa
\| \de \Xi\nnn \|_{\H^s_0(S_r)} \lesssim (\eps_0^\frac 23+\eps_{loc} M^2) \| \de\Psi\nn \|_{s-1}.
\eeaa
\item 
Using \eqref{eq:d1-d2-Thh-iterate}, we derive the equation for $\de\Thh\nnn$, which reads, schematically,
\beaa
 && \d_1\nnn\d_2\nnn (\de\Thh\nnn) \\
 &=&  \de\ah\nn \cdot \d_1\nn\d_2\nn \Thh\nn+  \de(\Kk_{\ell\leq 1}\nnn,\Kkd_{\ell\leq 1}\nnn)\\
 &&-(\d_1\nnn \d_2\nnn-\d_1\nn \d_2\nn)\Thh\nn \\
 &&+\nabh\nn ( (\thh\nn,\Thh\nn,\Pi\nn)\cdot \de (\slashed{d}(\log\ah\nnn), \de \Xi\nnn))+ Y\nn\cdot \de  \Xi\nnn\\
 &&+\nabh\nn ((\de \thh\nn\,\de \Thh\nn,\de \Pi\nn)\cdot ( \slashed{d}\log\ah\nnn,  \Xi\nnn))+\de Y\nn \cdot \Xi\nnn
\eeaa
Note that we have obtained the estimate for $\de u\nnn$, $\de\ah\nnn$ and $\de\Xi\nnn$. Hence, we deduce,
\beaa
 \| \de\Thh\nnn \|_{\H_0^{s}(S_r)} 
&\lesssim & \eps_0^\frac 23 \| \de\Psi\nn\|_{s-1}+r^{-1} M \| \de(\slashed{d} \log\ah)\nnn, \de\Xi\nnn\|_{\H_0^{s-1}(S_r)}\\
&& +\langle M \rangle \| \de \Psi\nn\|_{s-1} (\| \slashed{d}\log\ah\nnn,\Xi\nnn\|_{\H_0^{s-1}(S_r)})  \\
& \lesssim &  M (\eps_{loc}M^2 +\eps_0^\frac 23) \| \de\Psi\nn\|_{s-1}.
\eeaa
\item We now turn to the equation for $\de \Pi\nnn+\frac 12 \de\trt\nnn$, which, in view of \eqref{eq:iteration-Pi}-\eqref{eq:average-Pi}, reads
\begin{align*}
\laph\nnn (\de (\Pi+\frac 12\trt)\nnn) & = -(\laph\nnn-\laph\nn)(\de (\Pi+\frac 12 \trt)\nn)\\
&\quad + \de \Kk_{\ell\leq 1}\nnn-\overline{\de \Kk_{\ell\leq 1}\nnn}\nn +(\overline{\Kk_{\ell\leq 1}\nn}\nn-\overline{\Kk_{\ell\leq 1}\nn}\n),\\
\overline{\de (\Pi+\frac 12\trt)\nnn}\nnn &= - \overline{(\Pi+\frac 12\trt)\nnn}\nnn+\overline{(\Pi+\frac 12\trt)\nnn}\nn.
\end{align*}
This gives the estimate $\de \Pi\nnn+\frac 12 \de\trt\nnn$. Combining this with \eqref{eq:estimate-de-trt-de-c}, we obtain the estimate
\beaa
\| \de \Pi\nnn \|_{\H_0^s(S_r)} \lesssim M (\eps_{loc}M^2 +\eps_0^\frac 23) \| \de\Psi\nn\|_{s-1}.
\eeaa
\end{enumerate}
Therefore, in view of the definition of the norm in \eqref{eq:Psi-norms-gz-M}, in particular those with the factor $\langle M\rangle^{-1}$, we deduce the contraction estimate as stated.
\end{proof}
\begin{remark}
The uniqueness of the solution, as stated in Theorem \ref{thm:local-existence-precise}, follows similarly. Indeed, suppose we have two solutions $\Psi$, $\Phi$ on $(r_0,r_0+\eps_{loc}')\times \mathbb{S}^2$ with $\eps_{loc'}<\eps_{loc}$, we can consider their difference, and use a similar estimate to show that $\| \Psi-\Phi \|_{s-1}\leq C (\eps_{loc}'+\eps_0^\frac 23) \|\Psi-\Phi \|_{s-1}$, which yields $\Psi=\Phi$.
\end{remark}

\subsection{Identifying the limit}\lab{sec:identifying-the-limit}

In what follows, we refer to a quantity $\Psi$ as being unambiguous if the associated limiting quantity $\Psi\i$ obtained through the iteration is the same as the corresponding geometric quantity $\Psi(g\i)$ associated with $g\i$.
Also, we say an equation is unconditional if it holds regardless of whether $(g,k)$ solves the Einstein constraint equation. In particular, equations in Proposition \ref{prop:linearized-eqns} are all unconditional.
 
We first note that we define $g\i$ using $\ah\i$, $b\i$, and $u\i$ in view of \eqref{eq:iterate-ga}, \eqref{eq:iterate-g}:
 \bea\lab{eq:limiting-metric-g}
g\i:= (\ah\i)^2 dr^2 +  r^2 \left(e^{2u\i}\right) \gacr_{AB}\left(d\vth^A +(b\i)^A dr\right) \left(d\vth^B+(b\i)^B dr\right).
 \eea
 Therefore, by definition, $\ah\i$, $b\i$, and $u\i$ are unambiguous.
 The metric $g\i$ also determines the outward unit normal $N\i$ to the underlying $r$-spheres. We then define $k\i$ through its components in an orthonormal frame $\{N\i,e_a\i\}_{a=1,2}$:
\begin{equation}\lab{eq:def-k-infty}
\begin{gathered}
k\i(e_a\i,e_b\i) :=\Thh\i (e_a\i,e_b\i)+\frac 12 \trt\i \de_{ab},\\
k\i(N\i,e_a\i) := \Xi\i (e_a\i),\\
k\i(N\i,N\i) :=\Pi\i.
\end{gathered}
\end{equation}
Hence, the components of $k\i$ are also unambiguous.

\begin{proposition}\lab{prop:identifying-the-limits}
The constraint quantities $\CC_{Ham}$ and $\CC_{Mom}$ for $(g\i,k\i)$, as defined in \eqref{eq:def-CC-Mom}-\eqref{eq:def-CC-Ham}, vanish, i.e., $(g\i,k\i)$ solves the Einstein constraint equation \eqref{ece}.
\end{proposition}
\begin{proof}
Recall that in \cite{CK25}, the spherical metric $\ga\i$ is determined by integrating $\th\i$, and hence there is no ambiguity in $\th\i$. The main issue,  therefore, is to show that $K$ is unambiguous. In contrast, in the present setting, $K\i$ is directly tied to $u\i$, and the main issue is instead to show that $\trth$ is unambiguous.

 We proceed as follows:
 \begin{enumerate}
\item In view of the limiting version of \eqref{eq:iterate-laph-u-K}, i.e., $-\laph_{\gacr} u\i+1=(r^2\Kc\i) e^{2u\i}$, we deduce that $K\i=K(g\i)$ since, by \eqref{eq:laph-u-K-pert}, the latter is also determined by the same formula in terms of $u\i$. Moreover, the limiting version of \eqref{eq:iterate-ell=1-of-u} implies, on every $S_r$,
\begin{equation}\lab{eq:u-i-barycenter}
\int_{\mathbb{S}^2} e^{2u\i} x^i=0.
\end{equation}

\item The form of $g\i$ in \eqref{eq:limiting-metric-g} verifies the condition for \eqref{eq:nab-hot-b-flat} to hold. Therefore,\footnote{Recall again that we construct $b\i$ as a $1$-form.}
\beaa 
\nabz\hot b\i =-2 e^{-2u\i} \ah\i \, \thh(g\i).
\eeaa
On the other hand, the limiting version of \eqref{eq:iterate-nab-hot-b} implies $\nabz\hot b\i =-2 e^{-2u\i} \ah\i \, \thh\i$. Since $e^{-2u\i} \ah\i$ is close to $1$, this yields $\thh\i=\thh(g\i)$. 

\item Note the unconditional Codazzi equation \eqref{eq:unconditional-Codazzi-1} and the limiting version of \eqref{eq:unconditional-Codazzi-1-2} read, respectively,
\begin{align*}
\d_1\i \d_2\i \thh(g\i) &= \frac 12 \left(\laph\i \trth(g\i),0\right) -\d_1\i Y(g\i),\\
\d_1\i \d_2\i \thh\i &= \frac 12 \left(\laph\i \trth\i,0\right) - (\BB,\BBd) - \big(\Bb_{\ell\leq 1}\i,\Bbd_{\ell\leq 1}\i\big),
\end{align*}
Comparing these two equations, using $\thh\i=\thh(g\i)$ from the previous step, we obtain
\bea\lab{eq:relation-Y-ambiguity-trth-ambiguity}
-\d_1\i Y(g\i)+(\BB,\BBd)+\big(\Bb_{\ell\leq 1}\i,\Bbd_{\ell\leq 1}\i\big)=-\frac 12 \left(\laph\i \big(\trth(g\i)-\trth\i\big),0\right).
\eea

\item We use the notation
\beaa
&& \mu\i:= -\laph\i \left(\log\ah\i\right)+K\i -\frac 14\left(\trth\i \right)^2,\\
&& \mu(g\i):= -\laph\i \left(\log\ah\i\right) +K\i -\frac 14 \left(\trth(g\i)\right)^2.
\eeaa
Note that now $K\i$ and $\ah\i$ are unambiguous.
Taking the limit of \eqref{eq:iterate-laph-ao} (projected to $\ell\geq 1$), we deduce
\begin{equation}\lab{eq:mu-i-ell=0-mu-i}
\mu\i_{\ell\geq 1}=0,\quad \text{hence in particular $\mu\i_{\ell=0}=\mu\i_{\ell\neq 1}$}.
\end{equation}
The limiting version of \eqref{eq:R-transport-P2} then reads, using \eqref{eq:mu-i-ell=0-mu-i}, 
\begin{equation}\lab{eq:N-i-K-i}
\bsplit
N\i \Kc_{\ell\neq 1}\i&=  
 r^{-1} \mu\i_{\ell\neq 1} - 3r^{-1} \Kc\i_{\ell\neq 1}
-\BB-\Bb_{\ell=0}\i
-2r^{-3}\ao\i_{\ell\neq 1} \\
&\quad  +\left(\Ga_1\i\cdot \Ga_2\i+2\, \slashed{d}(\log\ah\i)\cdot Y\i\right)_{\ell\neq 1}+\RR_{tran}\i (\Kc\i) \\
&\quad  -\left(\thh\i\cdot \left(\nabh\hot\, \slashed{d}(\log\ah\i)-\slashed{d}(\log\ah\i)\hot\, \slashed{d}(\log\ah\i)\right)\right)_{\ell\neq 1},
\end{split}
\end{equation}
Comparing this with the unconditional equation \eqref{eq:N(Kc)}, noting again that $K\i=K(g\i)$, $\thh\i=\thh(g\i)$:
\begin{equation}\lab{eq:N-i-K-i-unconditional}
\bsplit
N\i \Kc_{\ell\neq 1}\i &= -3r^{-1} \Kc_{\ell\neq 1}\i-(\divh\i Y(g\i))_{\ell\neq 1} + r^{-1} \mu(g\i)_{\ell\neq 1} -2r^{-3} \ao\i_{\ell\neq 1}\\
&\quad  +\left(\Ga_1(g\i)\cdot \Ga_2(g\i)+2\, \slashed{d}(\log\ah\i)\cdot Y(g\i) \right)_{\ell\neq 1}+\RR_{tran}\i (\Kc\i) \\
&\quad  -\left(\thh\i \cdot \left(\nabh\hot\, \slashed{d}(\log\ah\i)-\slashed{d}(\log\ah\i)\hot\, \slashed{d}(\log\ah\i)\right)\right)_{\ell\neq 1},
\end{split}
\end{equation}
we deduce\footnote{Note that the coefficient of the $\RR_{tran}\i$ terms only depend on $\ah\i$, $b\i$, and $u\i$, which are all unambiguous. Therefore, the $\RR_{tran}\i(\Kc\i)$ in \eqref{eq:N-i-K-i} and \eqref{eq:N-i-K-i-unconditional} cancel.}
\begin{equation}\lab{eq:mu-ambiguity}
-r^{-1} \mu(g\i)_{\ell\neq 1}+ r^{-1} \mu\i_{\ell\neq 1} = \BB+\Bb_{\ell=0}\i-(\divh\i Y(g\i))_{\ell\neq 1} +\mathrm{Diff}_{NL},
\end{equation}
where $\mathrm{Diff}_{NL}$ denotes the following difference
\begin{equation*}
\mathrm{Diff}_{NL}:=\big(\Ga_1(g\i) \cdot \Ga_2(g\i)\big)_{\ell\neq 1} - \big(\Ga_1\i\cdot \Ga_2\i\big)_{\ell\neq 1}+2\slashed{d} (\log\ah\i)\cdot (Y(g\i)-Y\i).
\end{equation*}
Therefore, since now the remaining ambiguous quantity in $\Ga_i\i$ is $\thc\i$, the difference $\mathrm{Diff}_{NL}$ only consists of the following two schematic expressions 
\beaa
O(\eps) r^{-1} (Y(g\i)-Y\i),\quad O(\eps) r^{-2} (\trth(g\i)-\trth\i),
\eeaa
where $\eps:= \eps_0^\frac 23$ is the boundedness parameter for quantities in $\Ga_i\i$ obtained in Proposition \ref{prop:boundedness} (also recall the definition of the norm \eqref{eq:Psi-norms-gz-M}).

Now, recall the first component of 
\eqref{eq:relation-Y-ambiguity-trth-ambiguity} projected to $\ell\neq 1$:
\beaa
-(\divh \i Y(g\i))_{\ell\neq 1}+\BB + \Bb_{\ell=0}\i =-\frac 12 \left(\laph\i (\trth(g\i)-\trth\i) \right)_{\ell\neq 1}.
\eeaa
Plugging this to \eqref{eq:mu-ambiguity}, we obtain
\beaa
\frac 12 \left(\laph\i (\trth(g\i)-\trth\i) \right)_{\ell\neq 1}-r^{-1} (\mu(g\i)-\mu\i)_{\ell\neq 1}=\mathrm{Diff}_{NL}.
\eeaa
In view of the expressions of $\mu(g\i)$ and $\mu\i$, we see that this can be written as
\beaa
\left(\laph\i (\trth(g\i)-\trth\i) \right)_{\ell\neq 1}+2 r^{-2} (\trth(g\i)-\trth\i)_{\ell\neq 1}=\widetilde{\mathrm{Diff}}_{NL},
\eeaa
where  the error terms $\widetilde{\mathrm{Diff}}_{NL}$ have  the same schematic expression as $\mathrm{Diff}_{NL}$.
Since we are considering the $\ell\neq 1$ part, the elliptic estimates imply 
\bea\lab{eq:ambiguity-trth-Y}
\left \| \thc(g\i)_{\ell\neq 1}-\thc\i_{\ell\neq 1} \right\|_{\H_0^2(S_r)}\lesssim \eps \big\| r (Y(g\i)-Y\i) \big\|_{L_0^2(S_r)}.
\eea

\item Since \eqref{eq:u-i-barycenter} holds, we can apply \eqref{eq:divz-b-flat-ell=1-pert} to $g\i$ and obtain the relation
\beaa
\big(\divz (e^{2u\i} b\i)\big)_{\ell=1} & = -\left(e^{2u\i} (\ah\i\trth(g\i)-2r^{-1})\right)_{\ell=1}.
\eeaa
On the other hand, the limiting version of \eqref{eq:iterate-pa-r-u-ell=1} implies
\beaa
\big(\divz (e^{2u\i} b\i)\big)_{\ell=1} & = -\left(e^{2u\i} (\ah\i\trth\i-2r^{-1})\right)_{\ell=1}.
\eeaa
Therefore, comparing them yields $\thc\i_{\ell=1}- \thc(g\i)_{\ell=1}=O(\eps) (\thc\i-\thc(g\i))$. Combining this with \eqref{eq:ambiguity-trth-Y}, we obtain
\bea\lab{eq:ambiguity-trth-Y-with-ell=1}
\Big \| \thc(g\i)-\thc\i\Big\|_{\H_0^2(S_r)}\lesssim \eps \| r (Y(g\i)-Y\i)\|_{L_0^2(S_r)}.
\eea

\item 
The limiting version of \eqref{eq:iterate-Y} reads
\begin{align*}
\d_1\i Y\i =  (\BB, \BBd)+(\Bb_{\ell\leq 1}\i,\Bbd_{\ell\leq 1}\i)  -\overline{(\BB, \BBd)+(\Bb_{\ell\leq 1}\i,\Bbd_{\ell\leq 1}\i)}\i.
\end{align*}
Combining this with \eqref{eq:relation-Y-ambiguity-trth-ambiguity}, we deduce
\begin{align*}
\d_1\i (Y\i - Y(g\i)) = -\frac 12 \laph\i \left(\trth(g\i)-\trth\i\right) -\overline{(\BB, \BBd)+(\Bb_{\ell\leq 1}\i,\Bbd_{\ell\leq 1}\i)}\i.
\end{align*}
Integrating the equation over $S_r$ with respect to $\ga\i$, we see that the last term is zero. Then, plugging in the estimate \eqref{eq:ambiguity-trth-Y-with-ell=1}, we infer $\|Y\i-Y(g\i) \|_{\H^1(S_r)} \lesssim \eps \|Y\i-Y(g\i)\|_{L^2(S_r)}$. Hence, $Y(g\i)=Y\i$, which also implies $\trth(g\i)=\trth\i$.

\item We have now inferred that all quantities are unambiguous. Comparing the limiting version of \eqref{eq:R-transport-kac2} with the unconditional equation \eqref{eq:R-transport-kac} that contains $\CC_{Ham}(g\i,k\i)$, we obtain $\CC_{Ham}(g\i,k\i)=0$. Similarly, we compare \eqref{eq:N-trt-linearized} with the corresponding unconditional equations to deduce that $(\CC_{Mom}(g\i,k\i))_{N\i}=0$.

\item 
It remains to show that $\slashed\CC_{Mom}\i=0$.
Taking the limit of equations \eqref{eq:d1-d2-Thh-iterate}, \eqref{eq:nu-iterate}, \eqref{eq:curl-Xi-iterate-1}, \eqref{eq:curl-Xi-iterate-2}, and \eqref{eq:iteration-Pi}, we deduce the following relations:
\begin{align}
\lab{eq:d1-d2-Thh-iterate-limit}
\d_1\i \d_2\i (\ah\i \Thh\i) &= (\KK, -\KKd)+(\Kk_{\ell\leq 1}\i,\Kkd_{\ell\leq 1}\i)\\
&\notag\quad+\widetilde \Ga_1\i\cdot \widetilde \Ga_2\i,\\
\lab{eq:div-Xi-limit}
\divh\i\Xi\i &=0,\\
r^{-4} \ah\i N\i (r^4 (\curlh\i \Xi\i+\overline{c\i}\i)) &= \KKd-\Kkd_{\ell\leq 1}\i, \\
\laph\i (\Pi\i+\frac 12 \trt\i) &=\KK+\Kk_{\ell\leq 1}\i -\overline{\KK+\Kk_{\ell\leq 1}\i}\i.
\end{align}
Since $\ah\i N\i (r)=1$, we see that $r^{-4} \ah\i N\i(r^4 \overline{c\i}\i)$ is only dependent on $r$. Denoting it by $F(r)$, we obtain
\bea\lab{eq:curl-Xi-limit-correction}
r^{-4} \ah\i N\i (r^4 (\curlh\i \Xi\i))= \KKd-\Kkd_{\ell\leq 1}\nn-F(r).
\eea
Meanwhile, the unconditional equations \eqref{eq:N-div-Xi}-\eqref{eq:N-curl-Xi}, applied to $(g\i,k\i)$, read
\begin{equation}\lab{eq:N-div-Xi-uncond-i}
\bsplit
(\ah\i N\i +4r^{-1}) \divh\i\Xi \i &= -\divh\i\divh\i (\ah\i \kh\i) \\
&\quad + \laph\i (\frac 12 \trt\i+\Pi\i) +\widetilde \Ga_1\i\cdot \widetilde \Ga_2\i\\
&\quad +\divh\i (\ah\i \slashed\CC\i_{Mom}),
\end{split}
\end{equation}
and
\begin{equation}
\lab{eq:N-curl-Xi-uncond-i}
\bsplit
r^{-4} \ah\i N\i (r^4 \curlh\i \Xi\i) &= - \curlh\i\divh\i (\ah\i \kh\i)+\widetilde \Ga_1\i\cdot \widetilde \Ga_2\i\\
&\quad +\curlh\i (\ah\i \slashed\CC_{Mom}\i).
\end{split}
\end{equation}
Therefore, plugging \eqref{eq:d1-d2-Thh-iterate-limit} to the right-hand sides, and \eqref{eq:div-Xi-limit}, \eqref{eq:curl-Xi-limit-correction} to the left-hand sides of \eqref{eq:N-div-Xi-uncond-i}, \eqref{eq:N-curl-Xi-uncond-i}, we obtain
\begin{align}
0 &= -\overline{\KK+\Kk_{\ell\leq 1}\i}\i +\divh\i (\ah\i \slashed\CC_{Mom}\i),\\
-F(r) &= \curlh\i (\ah\i \slashed\CC_{Mom}\i).
\end{align}
Integrating over spheres with respect to $\ga\i$, we deduce $\overline{\KK+\Kk_{\ell\leq 1}\i}\i=F(r)=0$. Hence $\divh\i (\ah\i \slashed\CC_{Mom}\i)=\curlh\i (\ah\i \slashed\CC_{Mom}\i)=0$, which yields $\slashed\CC_{Mom}\i=0$. This finishes the proof of the fact that $(g\i,k\i)$ solves the Einstein constraint equation \eqref{ece}. 
\end{enumerate}
We have shown in the proof that $\mu_{\ell\geq 1}\i=\nu\i=0$, $\displaystyle \int_{\mathbb{S}^2} e^{2u\i} x^i=0$, and
\begin{align*}
\d_1\i Y\i & =(\BB,\BBd)+(\Bb_{\ell\leq 1}\i, \Bbd_{\ell\leq 1}\i),\\
r^{-4} \ah\i N\i (r^4 (\curlh\i \Xi\i)) &= \KKd-\Kkd_{\ell\leq 1}\i, \\
\laph\i (\Pi\i+\frac 12 \trt\i) &=\KK+\Kk_{\ell\leq 1}\i .
\end{align*}
It is also straightforward to verify that the remaining gauge conditions hold, by taking the limit of equations \eqref{eq:spherical-mean-ao2}, \eqref{eq:average-Pi}, and \eqref{eq:iterate-ell=1-of-curl-b}:
\beaa
\overline{\ao\i}\i=0, 
\quad \overline{\Pi\i}\i=-\frac 12\overline{\trt\i}\i,\quad (e^{2u\i}\curlh\i b\i)_{\ell=1}=0.
\eeaa
This completes the verification of all properties listed in Theorem \ref{thm:local-existence-precise}.
\end{proof}

\section{Precise Statement and Proof of the Global Result}\lab{sec:global-result}

\subsection{Scalar model}

We first consider the scalar model
\begin{equation}\lab{eq:scalar-model}
\pa_r \phi+\la r^{-1}\phi=F+\phi^2,\qquad \phi|_{r=1}\sim \eps,
\end{equation}
where $|F|\lesssim \eps r^{-2}$. Our target is   the bound $|\phi|\lesssim \eps r^{-1}$.
Passing to $|\phi|$ and multiplying by $r^\la$, we get
\beaa
\pa_r(r^\la |\phi|)\lesssim r^\la |F|+r^\la |\phi|^2.
\eeaa
If we assume $|\phi|\lesssim \eps r^{-1}$, then both terms on the right are $\lesssim \eps r^{\la-2}$, and hence
\beaa
r^\la |\phi|(r)\lesssim \eps+\eps\int_1^r r'^{\la-2}\,dr'.
\eeaa
Now the issue is whether dividing by $r^\la$ recovers the desired $r^{-1}$ decay. If $\la<1$, the integral is bounded, but one only gets $|\phi|\lesssim \eps r^{-\la}$. If $\la=1$, one suffers from a logarithmic loss $|\phi|\lesssim \eps r^{-1}\log r$. If $\la>1$, the integral is $\lesssim r^{\la-1}$, which gives exactly $|\phi|\lesssim \eps r^{-1}$. Thus the $r^{-1}$ bound is propagated precisely when $\la>1$.

\begin{remark}\lab{rem:la'-la}
The same conclusion can be deduced if one uses any lower weight $r^{\la'}$ with $1<\la'<\la$. Indeed,
\beaa
\pa_r(r^{\la'}|\phi|)+(\la-\la')r^{\la'-1}|\phi|
\lesssim r^{\la'}|F|+r^{\la'}|\phi|^2,
\eeaa
so the extra term has a favorable sign. Under the same bound $|\phi|\lesssim \eps r^{-1}$, the right-hand side is $\lesssim \eps r^{\la'-2}$, and integrating again yields $r^{\la'}|\phi|\lesssim \eps r^{\la'-1}$, hence $|\phi|\lesssim \eps r^{-1}$. This observation will be useful later in the matrix setting.
\end{remark}

\begin{remark}
[The behavior of $\trt$]
According to \eqref{eq:N-trt-linearized}, the transport equation of $\trt$ reads
\beaa
\nabh_N \trt= \divh \Xi +2r^{-1} \Pi -r^{-1} \trt+\cdots.
\eeaa
In \cite{CK25}, the free scalar $\KK$ was chosen to be $\laph \Pi$ (at the linear $\ell\geq 2$ level). In the current situation, however, this would result in integrating\footnote{This corresponds  to the case  $\la=1$, which  we want to avoid in \eqref{eq:scalar-model}.} $\nabh_N \trt =-r^{-1} \trt +\cdots$. This is why we had to  modify the condition for $\KK$ in \eqref{eq:prescribed-scalar-conditions}, compared with \eqref{eq:prescribed-scalar-CK25} imposed in \cite{CK25}. 
Similar modification is made for $\ell=0$ conditions in \eqref{eq:ell=1-condition-u-solution1}. 
\end{remark}

\subsection{Precise statement of   Theorem \ref{thm:main-intro}}

In this section, we denote for simplicity $\H^s:= \H^s_\ga$, where $\ga$ is the induced metric of the solution metric $g$ on $S_r$. Moreover, from now on, we identify $\widetilde\Ga$ with $\Ga$.

\begin{theorem}[Precise statement of the main theorem]
\lab{main-thm}
Fix a positive integer $s\geq 3$ and a small constant $\eps_0>0$. For given $O_{s+2}(\eps_0)$-almost round canonical sphere data $(S,u,\vth^A,\trth,\trt,\Xi)$ of radius $r_0$, as defined in Definition \ref{eq:almost-round-sphere-data}, and prescribed scalars $(\BB,\BBd, \KK,\KKd)$ on $(r_0,\infty)\times \mathbb{S}^2$ satisfying the bounds
\beaa
\sup_{r\in (r_0,\infty)} r^2 \| \BB,\BBd,\KK,\KKd \|_{\H^s(S_r)} \leq \eps_0 ,
\eeaa
there exists a unique solution to \eqref{ece} on $(r_0,\infty)\times \mathbb{S}^2$ satisfying the conditions listed in  the local existence result of Theorem \ref{thm:local-existence-precise}. Moreover, for quantities defined in Section \ref{sec:geometric-quantities}, we have the estimate
\beaa
|\ao, b, u|\lesssim \eps_0,\quad |\thh, \thc, \trt, \Thh, \Xi, \Pi |\lesssim \eps_0 r^{-1},\quad |\Kc, Y|\lesssim \eps_0 r^{-2}.
\eeaa
\end{theorem}

\begin{remark}
It is straightforward to adapt the proof to obtain the minimal decay case \eqref{eq:minimal-decay-intro}, namely the case with $O(\eps_0 r^{-\de})$ decay at the metric level, for some $0<\de<1$, instead of the borderline $O(\eps_0)$ case considered here. Indeed, it suffices to impose the corresponding condition for the scalars $(\BB,\BBd,\KK,\KKd)$, and this only makes the nonlinear terms easier to control.
\end{remark}

\begin{remark}
Theorem \ref{main-thm}, combined with the local existence result in Theorem \ref{thm:local-existence-precise} with large free data, can also be applied to significantly generalize the result of Li--Yu \cite{LiYu} on the construction of complete, short pulse,  Cauchy data, see \cite{CK26-2}.
\end{remark}

\subsection{Bootstrap argument}
We introduce the abstract notation
\begin{align*}
\Psi=\left(\thc, \Kc, \ao, \thh, Y, u, b,
\trt, \kh, \Xi, \Pi, 
(\Bb_{\ell\leq 1},\Bbd_{\ell\leq 1}), (\Kk_{\ell\leq 1},\Kkd_{\ell\leq 1})\right),
\end{align*}
and define the following norm:
\begin{equation}
\bsplit
\| \Psi \|_{s,\ga,  I} & := \sup_{r\in I}\, r^{-1} \| \ao, b, u \|_{\H^{s+2}_\ga(S_r)}+ \| \thc, \thh, \trt, \Thh, \Xi, \Pi \|_{\H^{s+1}_\ga(S_r)} + r \| \Kc, Y \|_{\H_\ga^s(S_r)} \\
&\quad + r^3 |\Bb_{\ell\leq 1},\Bbd_{\ell\leq 1}, \Kk_{\ell\leq 1},\Kkd_{\ell\leq 1} |,
\end{split}
\end{equation}
where   the $\H_\ga^s$ norms on the  $S_r$ sphere have  been introduced in  Definition \ref{def:L2-Hs-S_r}.

We run the following bootstrap argument. 

{\bf BA.}  Assume  that $r^*$ is a finite maximal value of $r$ such that the following estimates hold on $I=[r_0,r]$:
\beaa
\| \Psi \|_{s, \ga,I} \leq \eps,\qquad \eps:= \eps_0^\frac 23.
\eeaa
Note that this in particular implies, using the local existence (with $r_0$ replaced by $r_*$), that the solution can be extended beyond $r=r^*$.
We show in this section that this implies
\beaa
\| \Psi \|_{s,\ga, I} \leq \frac 12 \eps,
\eeaa
and hence, by continuity of the norm, contradicts the maximality of $r^*$. Since the local existence applied to $r_0$ already ensures the existence of $r^*$, this yields $r^*=\infty$.

In view of  the local existence statement in Theorem \ref{thm:local-existence-precise},  the HCS  system  of Proposition \ref{prop:linearized-eqns} can be recast by using the   gauge conditions \eqref{eq:ell=1-condition-u-solution1}-\eqref{eq:ell=1-condition-u-solution3}, the  free scalars conditions \eqref{eq:prescribed-scalar-conditions}, and the fact that  $\CC_{Ham}=0$, $\CC_{Mom}=0$. Thus, in particular, we will make use of the following\footnote{The remaining equation will be referred to whenever needed.}
\begin{align}
\lab{eq:N-thc-estimate-part}
N \thc &= \mu_{\ell=0} -2r^{-1}\thc-2 r^{-2}\ao +\Ga_1\cdot \Ga_1, \\
\lab{eq:N-K-ell-neq-1-estimate-part}
N \Kc_{\ell\neq 1} &= r^{-1} \mu_{\ell=0}- (\divh Y)_{\ell\neq 1} -3r^{-1}\Kc_{\ell\neq 1} -2 r^{-3} \ao_{\ell\neq 1} \\
&\notag\quad +\Ga_1\cdot \Ga_2+\RR_{tran}(\Kc),\\
\lab{eq:laph-ao-estimate-part}
(\laph (\log\ah))_{\ell\geq 1}&= (\Kc- r^{-1} \thc +\Ga_1\cdot\Ga_1)_{\ell\geq 1},\\
\lab{eq:ao-ell=0-estimate-part}
\overline{\ao} &=0,\\
\lab{eq:laph-u-estimate-part}
-\laph_{\gacr} u-2u &=r^2 \Kc+r^2\Kc \cdot u,\\
\lab{eq:u-ell=1-estimate-part}
u_{\ell=1} &=u\cdot u.
\end{align}
Note that the system \eqref{eq:N-thc-estimate-part}-\eqref{eq:u-ell=1-estimate-part}   for $(\thc, \Kc, \ao,u)$ is self-contained at the linear level, in view of the fact that  $(\divh Y)_{\ell\neq 1}$ is  prescribed   by $\BB$, according to  \eqref{eq:prescribed-scalar-conditions}.  

\subsubsection{The main coupled  system  \eqref{eq:N-thc-estimate-part}-\eqref{eq:u-ell=1-estimate-part}}
We first control the $\ell=1$ part of the Gauss curvature.
\begin{proposition}\lab{prop:estimate-K-ell=1}
We have
\begin{align*}
r^2 | \Kc_{\ell=1} | \lesssim \eps^2,\quad r^{-1} \| u \|_{\H^{s+2}(S_r)} \lesssim  \eps^2+ r\| \Kc_{\ell\neq 1} \|_{\H^s(S_r)}.
\end{align*}
\end{proposition}
\begin{proof}
Projecting \eqref{eq:laph-u-estimate-part} to $\ell=1$ modes, we obtain the estimate $|r^2\Kc_{\ell=1} |\lesssim \eps^2$. Then, applying the standard elliptic estimate to \eqref{eq:laph-u-estimate-part} with the condition \eqref{eq:u-ell=1-estimate-part} on $u_{\ell=1}$, we obtain the bound for $u$ as stated.
\end{proof}

We now derive the main coupled system for $(\thc_{\ell\geq 2},\Kc_{\ell\geq 2})$.
\begin{lemma}
The following system holds:
\begin{equation}\lab{eq:ell-geq-2-system-thc-Kc}
\bsplit
N \begin{pmatrix} \thc_{\ell\geq 2} \\ r\Kc_{\ell\geq 2} \end{pmatrix} &= r^{-1} \begin{pmatrix} -2+2(r^2 \laph)^{-1} & -2 (r^2\laph)^{-1} \\ 2 (r^2\laph)^{-1} & -2-2(r^2\laph)^{-1} \end{pmatrix} \begin{pmatrix} \thc_{\ell\geq 2} \\ r\Kc_{\ell\geq 2} \end{pmatrix}+\begin{pmatrix} 0 \\ r\BB \end{pmatrix}\\
&\quad + \Ga_1\cdot r\Ga_2+\RR(\thc,r\Kc),
\end{split}
\end{equation}
where the $\RR$ notation is introduced in Definition \ref{def:error-terms}.
\end{lemma}
\begin{proof}

Inverting the Laplacian in \eqref{eq:laph-ao-estimate-part}, we obtain
\bea\lab{eq:log-ah-substitute}
\log\ah-\overline{\log\ah}=\laph^{-1} (\Kc-r^{-1} \thc+\Ga_1\cdot \Ga_1),
\eea
where we extend the definition $\laph^{-1}$ to any function $\phi$ by defining $\laph^{-1} \phi:= \laph^{-1} (\phi-\overline{\phi})$.
Plugging \eqref{eq:log-ah-substitute} into \eqref{eq:N-thc-estimate-part} and \eqref{eq:N-K-ell-neq-1-estimate-part}, using $(\divh Y)_{\ell\geq 2}=\BB$, we deduce
\begin{align}
\lab{eq:N-thc-estimate-part-substituted}
N \thc_{\ell\geq 2} &=  -2r^{-1}\thc_{\ell\geq 2}-2 r^{-2}\left(\laph^{-1} (\Kc-r^{-1} \thc+\Ga_1\cdot \Ga_1) \right)_{\ell\geq 2}+\Ga_1\cdot \Ga_1+\RR_{tran}(\thc), \\
N \Kc_{\ell\geq 2} &= - \BB -3r^{-1}\Kc_{\ell\geq 2} -2 r^{-3} \left(\laph^{-1} (\Kc-r^{-1} \thc+\Ga_1\cdot \Ga_1)\right)_{\ell\geq 2} +\Ga_1\cdot \Ga_2+\RR_{tran}(\Kc),
\end{align}
where the $\RR_{tran}$ notation is introduced in Definition \ref{def:error-terms}, according to which we also write
\beaa
(\laph^{-1} \phi)_{\ell\geq 2}=\laph^{-1} (\phi_{\ell\geq 2}) + \RR_{ellip}(\phi).
\eeaa
Therefore, we obtain
\beaa
N \begin{pmatrix} \thc_{\ell\geq 2} \\ \Kc_{\ell\geq 2} \end{pmatrix} = r^{-1} \begin{pmatrix} -2+2r^{-2} \laph^{-1} & -2 r^{-1} \laph^{-1} \\ 2r^{-3}\laph^{-1} & -3-2r^{-2}\laph^{-1} \end{pmatrix} \begin{pmatrix} \thc_{\ell\geq 2} \\ \Kc_{\ell\geq 2} \end{pmatrix} +\begin{pmatrix} 0 \\ r\BB \end{pmatrix}+\begin{pmatrix} \Ga_1\cdot \Ga_1 +\RR(\thc) \\ \Ga_1\cdot \Ga_2+\RR(\Kc) \end{pmatrix},
\eeaa
where $\RR$ includes both $\RR_{tran}$ and $\RR_{ellip}$, which can be controlled using Lemma \ref{lem:RR-transport} and Lemma \ref{lem:comparison}.
Then, incorporating $\Kc_{\ell\geq 2}$ with an $r$-weight, we obtain \eqref{eq:ell-geq-2-system-thc-Kc} as required.
\end{proof}

As in the case of  the scalar model, to integrate \eqref{eq:ell-geq-2-system-thc-Kc} forward in $r$, we need to incorporate suitable $r$-weights. We  rely on the following lemma.
\begin{lemma}\lab{lem:matrix-operator-sign}
Suppose $L\colon H\to H$ is a self-adjoint bounded linear operator with $\|L\|\le C$ on a Hilbert space $H$. For any given $\de>0$, the matrix operator
\beaa
A_\de[L]:=
\begin{pmatrix}
-\de-L & L\\
-L & -\de+L
\end{pmatrix}
\eeaa
is negative for some non-degenerate inner product $\langle\cdot,\cdot\rangle$ on $H\times H$, i.e., for any $(x,y)\in H\times H$, we have
\begin{equation*}
\Big\langle A_\de[L] \binom{x}{y}, \binom{x}{y} \Big\rangle + \Big\langle \binom{x}{y}, A_\de[L] \binom{x}{y} \Big\rangle \leq 0.
\end{equation*}
\end{lemma}

\begin{proof}
Write $A_\de[L]=-\de I+Q$, with
$Q:=
\begin{pmatrix}
-L&L\\
-L&L
\end{pmatrix}$. Consider
\beaa
G_\de:=\begin{pmatrix}
1&-1+\eta^2\\
-1+\eta^2&1
\end{pmatrix}
\eeaa
for some $\eta\in (0,1)$ to be determined.
Since the eigenvalues of $G_\de$ are $\eta^2$ and $2-\eta^2$, this defines an inner product on $H\times H$.
Using that $L$ is self-adjoint, one computes
\beaa
G_\de Q+Q^*G_\de=
\begin{pmatrix}
-2\eta^2L&0\\
0&2\eta^2L
\end{pmatrix}.
\eeaa
Hence
\beaa
M_\de:= \frac12\big(G_\de A_\de[L]+A_\de[L]^*G_\de\big)=
\begin{pmatrix}
-\de-\eta^2L&\de(1-\eta^2)\\
\de(1-\eta^2)&-\de+\eta^2L
\end{pmatrix}.
\eeaa
It suffices to show $M_\de<0$.
Now conjugate by
\beaa
U:=\frac1{\sqrt2}
\begin{pmatrix}
I&I\\
I&-I
\end{pmatrix}.
\eeaa
A direct computation gives
\beaa
U^*M_\de U
=
\begin{pmatrix}
-\de\eta^2I&-\eta^2L\\
-\eta^2L&-\de(2-\eta^2)I
\end{pmatrix}.
\eeaa
Thus it suffices to show that the above operator is negative.
Let $(x,y)\in H\times H$. Then
\beaa
\Big\langle U^*M_\de U \binom{x}{y}, \binom{x}{y}\Big\rangle =
-\de\eta^2\|x\|^2-\de(2-\eta^2)\|y\|^2-2\eta^2\langle Lx,y\rangle.
\eeaa
Using $\|L\|\le C$, it is straightforward to show that this is negative for $\eta:=\min\Big\{1,\frac{\de}{2C}\Big\}$.
\end{proof}

\begin{proposition}
We have the improved estimate
\bea\lab{eq:improved-estimate-thc-Kc}
\| \thc \|_{\H^{s+1}(S_r)} \lesssim \eps_0+ \eps^2,\quad r\|\Kc \|_{\H^s(S_r)} \lesssim \eps_0+ \eps^2.
\eea
As a direct consequence in view of Proposition \ref{prop:estimate-K-ell=1}, we also obtain the improved estimate for the conformal factor $r^{-1} \| u \|_{\H^{s+2}(S_r)} \lesssim \eps_0+ \eps^2$.
\end{proposition}
\begin{proof} We proceed  in steps as follows:

{\bf Step 1.}
 Consider the coupled system for $\ell\geq 2$ in view of \eqref{eq:ell-geq-2-system-thc-Kc}:
\begin{equation}\lab{eq:N-de-weight-coupled}
\bsplit
N \begin{pmatrix} r^{2-\de} \thc_{\ell\geq 2} \\ r^{3-\de} \Kc_{\ell\geq 2}\end{pmatrix} &= r^{-1} \begin{pmatrix} -\de+2(r^2 \laph)^{-1} & -2 (r^2\laph)^{-1} \\ 2 (r^2\laph)^{-1} & -\de-2(r^2\laph)^{-1} \end{pmatrix} \begin{pmatrix} r^{2-\de} \thc_{\ell\geq 2} \\ r^{3-\de} \Kc_{\ell\geq 2} \end{pmatrix} \\
&\quad +\begin{pmatrix} 0 \\ r^{3-\de} \BB \end{pmatrix}+ r^{2-\de} \Ga_1\cdot r \Ga_2+r^{2-\de} \RR(\thc,r\Kc).
\end{split}
\end{equation}
The matrix operator on the right is exactly $A_\de[(r^2\laph)^{-1}]$ as considered in Lemma \ref{lem:matrix-operator-sign}, with $L=r^2 \laph^{-1}$ uniformly bounded for all $r$. Applying Lemma \ref{lem:matrix-operator-sign}, we see that the operator is negative with respect to the inner product $\H^s(S_r) \otimes G$, where $G$ is a positive definite inner product on $\mathbb{R}^2$ independent of $r$. This allows us to integrate in the forward direction using the second estimate in Lemma \ref{lem:integration-lemma-1} with $\la=0$, which yields
\begin{align*}
r^{-1} \| r^{2-\de} \thc_{\ell\geq 2}, r^{3-\de} \Kc_{\ell\geq 2} \|_{\H^s(S_r)} &\lesssim \eps_0+ \int_{r_0}^r r'^{2-\de} (r'^{-1} \|\Ga_1\cdot r\Ga_2\|_{\H^s(S_{r'})})+ r'^{3-\de}\cdot r'^{-1} \|\BB\|_{\H^s(S_{r'})}\, dr' \\
&\lesssim \eps_0+ \int_{r_0}^r r'^{2-\de} \cdot (\eps^2+\eps_0) r'^{-2} dr' \lesssim \eps_0+(\eps^2+\eps_0) r^{1-\de},
\end{align*}
i.e., 
\beaa
\|  \thc_{\ell\geq 2} \|_{\H^s(S_r)} \lesssim \eps_0 + \eps^2,\quad r\|\Kc_{\ell\geq 2}\|_{\H^s(S_r)} \lesssim \eps_0+ \eps^2.
\eeaa
Moreover, combining this with the estimate of $\Kc_{\ell=1}$ obtained in Proposition \ref{prop:estimate-K-ell=1}, we deduce the estimate for $\Kc_{\ell\geq 1}$.

{\bf Step 2.}
In view of the condition\footnote{
The $\ell=0$ part of $\ao$ corresponds to the freedom of relabeling a given sphere foliation. In the current scenario, using $u=O(\eps)$, we have $\mathrm{Area}(S_r)=r^2(1+O(\eps))$, and it is in fact not possible to make $r$ exactly the area radius,  which at the linear level would correspond to the relation $\overline{\ao}=-\frac 12 r \overline{\thc}$.} \eqref{eq:ao-ell=0-estimate-part} and the definition of $\mu$ in \eqref{eq:def-mu}, when projected to $\ell=0$, the equations \eqref{eq:N-thc-estimate-part}-\eqref{eq:N-K-ell-neq-1-estimate-part} read
\begin{align}
\lab{eq:N-thc-estimate-part-ell=0}
N \thc_{\ell=0} &= \Kc_{\ell=0} -3r^{-1}\thc_{\ell=0} +\Ga_1\cdot \Ga_1+\RR(\thc,r^{-1}\ao), \\
\lab{eq:N-K-ell-neq-1-estimate-part-ell=0}
N \Kc_{\ell=0} &= -r^{-2}\thc_{\ell=0} -(\divh Y)_{\ell=0} -2r^{-1}\Kc_{\ell=0}  +\Ga_1\cdot \Ga_2+\RR(\Kc,r^{-2}\ao).
\end{align}
Clearly, $(\divh Y)_{\ell=0}=O(\eps r^{-1}) \cdot Y$. We now consider the weighted variables $r^{2-\de} \thc_{\ell=0}$ and $r^{3-\de} \Kc_{\ell=0}$, whose equations read
\begin{align*}
\pa_r \begin{pmatrix} r^{2-\de} \thc_{\ell=0} \\ r^{3-\de} \Kc_{\ell=0} \end{pmatrix}=r^{-1} \begin{pmatrix} -1-\de & 1 \\ -1 & 1-\de \end{pmatrix} \begin{pmatrix} r^{2-\de} \thc_{\ell=0} \\ r^{3-\de} \Kc_{\ell=0} \end{pmatrix} + r^{3-\de} (\Ga_1\cdot \Ga_2+\RR(\Kc,r^{-1}\thc,r^{-2}\ao)).
\end{align*}
The matrix has $\det=-1+\de^2+1=\de^2 >0$, $\tr=-2\de<0$, hence two negative eigenvalues. This allows us to use Lemma \ref{lem:integration-lemma-1} to integrate forward in $r$ and deduce
\beaa
|r^{2-\de} \thc_{\ell=0},r^{3-\de} \Kc_{\ell=0}| \lesssim \eps_0+ \int_{r_0}^r r'^{3-\de} \cdot \eps^2 r'^{-3}\, dr' \lesssim \eps_0+ \eps^2 r^{1-\de},
\eeaa
which proves the improved estimate for $\thc_{\ell=0}$ and $\Kc_{\ell=0}$.

{\bf Step 3.}
Plugging \eqref{eq:log-ah-substitute} again into \eqref{eq:N-thc-estimate-part}, this time projected to $\ell=1$, we obtain
\beaa
N \thc_{\ell=1} =  -2r^{-1}\thc_{\ell=1}-2 r^{-2}\laph^{-1} (\Kc_{\ell=1}-r^{-1} \thc_{\ell=1}+\Ga_1\cdot \Ga_1) +\Ga_1\cdot \Ga_1+\RR_{tran}(\thc), 
\eeaa
which can be rewritten as
\beaa
N \thc_{\ell=1} =  -3r^{-1}\thc_{\ell=1}-2 r^{-2}\laph^{-1} (\Kc_{\ell=1}) +\Ga_1\cdot \Ga_1+ \RR_{tran} (\thc).
\eeaa
Applying the second estimate in Lemma \ref{lem:integration-lemma-1} with $\la=3$, we obtain the improved the estimate for $\thc_{\ell=1}$.\footnote{We note that the $\ell=1$ equations for $\thc$, $\trt$, and $\curlh \Xi$ are similar to those derived in \cite[Section 2.6]{CK25}. In particular, they are rigid in the sense that no free scalars enter as inhomogeneous terms; consequently, it suffices to impose the initial condition at $\{r=r_0\}$ (or $r=\infty$ in the context of \cite{CK25}).}

This finishes the $\H^s$ estimates for $\thc$ and $\Kc$. Then, applying elliptic estimates to \eqref{eq:laph-ao-estimate-part}, \eqref{eq:ao-ell=0-estimate-part}, we obtain the improved estimate for $\ao=\ah-1$:
\beaa
r^{-1} \| \ao \|_{\H^{s+2}(S_r)} \lesssim \eps_0+ \eps^2.
\eeaa 
To close the estimates in \eqref{eq:improved-estimate-thc-Kc}, we commute the equation \eqref{eq:N-thc-estimate-part} with $(r\nabh)^{s+1}$ and integrate in $r$. Since we now have enough control of $\ah$, we obtain the $\H^{s+1}$ estimate for $\trth$.
\end{proof}

\subsubsection{Bounds for the remaining quantities}
We now examine the remaining equations, which, similar to \eqref{eq:N-thc-estimate-part}-\eqref{eq:u-ell=1-estimate-part}, follow by combining the HCS  equations  of Proposition \ref{prop:linearized-eqns} with the conditions in Theorem \ref{thm:local-existence-precise}:
\begin{align}
\lab{eq:unconditional-Codazzi-1-2-estimate}
\d_1 \d_2 \thh &= \big(\frac 12 \laph\thc,0\big)-(\BB+\Bb_{\ell\leq 1},\BBd+\Bbd_{\ell\leq 1}),\\
\lab{eq:estimate-Y}
\d_1 Y &=  (\BB, \BBd)+(\Bb_{\ell\leq 1},\Bbd_{\ell\leq 1})  ,
\\
\lab{eq:N-trt-estimate}
N\trt &= 2 r^{-1} \Pi  - r^{-1} \trt +\Ga_1\cdot \Ga_1 ,\\
\lab{eq:d1-d2-Thh-estimate}
\d_1 \d_2 (\ah \Thh) &= (\KK, -\KKd)+(\Kk_{\ell\leq 1},\Kkd_{\ell\leq 1})+\Ga_1\cdot \Ga_2,\\
\lab{eq:nu-estimate}
\divh \Xi &=0, \\
\lab{eq:curl-Xi-estimate}
\ah N \curlh \Xi  &=-4r^{-1} \curlh \Xi+ \KKd,\\
\lab{eq:estimate-Pi}
\laph \left(\Pi +\frac 12 \trt \right) &= \KK + \Kk_{\ell\leq 1},\\
\lab{eq:average-Pi-estimate}
\overline{\Pi}&= -\frac 12 \overline{\trt},\\ 
\lab{eq:nab-hot-b-estimate}
\nabz\hot b^{\flat,0} &=-2 e^{-2u} \ah\, \thh,\\
\lab{eq:div-b-ell=1}
(\divz (e^{2u} b^{\flat,0}))_{\ell=1} & = -\left(e^{2u} (\ah\, \trth -2r^{-1})\right)_{\ell=1},\\
\lab{eq:curl-b-ell=1}
(\curlz (e^{2u} b^{\flat,0}))_{\ell=1}&=0.
\end{align}
\begin{proposition}
We have 
\beaa
\| \thh ,rY, \trt, \Thh, \Xi, \Pi \|_{\H^{s+1}(S_r)} \lesssim \eps_0+\eps^2.
\eeaa
\end{proposition}
\begin{proof}
We proceed in steps as follows: 
\begin{enumerate}
\item 
To obtain the estimate for $\Bb_{\ell\leq 1},\Bbd_{\ell\leq 1}$, we apply Lemma \ref{lem:comparison} to \eqref{eq:unconditional-Codazzi-1-2-estimate}. This yields 
\beaa
r^2 |\Bb_{\ell\leq 1},\Bbd_{\ell\leq 1} |\lesssim r|\trth_{\ell=1}|+ \eps^2 \lesssim \eps_0 + \eps^2,
\eeaa
where we also used the improved bound for $\trth$ obtained in \eqref{eq:improved-estimate-thc-Kc}. Then, we use the Hodge estimate in Lemma \ref{lem:comparison} to infer 
\beaa
\|\thh\|_{\H^{s+1}(S_r)} \lesssim \|\thc \|_{\H^{s+1}(S_r)}+ \eps^2 \lesssim \eps_0 + \eps^2.
\eeaa
\item Applying the same argument to \eqref{eq:d1-d2-Thh-estimate}, we obtain 
\beaa
|\Kk_{\ell\leq 1},\Kkd_{\ell\leq 1}|\lesssim \eps_0+ \eps^2 ,\quad \|\Thh\|_{\H^{s+1}(S_r)}\lesssim \eps_0+ \eps^2.
\eeaa
\item Applying elliptic estimates to \eqref{eq:estimate-Pi}-\eqref{eq:average-Pi-estimate}, we obtain 
\bea\lab{eq:improved-bound-Pi-trt-sum}
\|\Pi+\frac 12 \trt\|_{\H^{s+1}(S_r)}\lesssim  \eps_0+ \eps^2.
\eea
\item Using \eqref{eq:nab-hot-b-estimate}
-\eqref{eq:curl-b-ell=1} and the improved estimate for $\thh$, $\thc$ and $\ao$, we  obtain the improved estimate for $b$: 
\beaa
r^{-1}\| b \|_{\H^{s+2}(S_r)} \lesssim \| \thh \|_{\H^{s+1}(S_r)}+|r\thc_{\ell=1},\ao_{\ell=1}| + \eps^2 \lesssim \eps_0 + \eps^2.
\eeaa
\item We write 
\eqref{eq:N-trt-estimate} as
\beaa
N\trt = -2r^{-1} \trt +2r^{-1} (\Pi+\frac 12 \trt) + \Ga_1\cdot \Ga_1.
\eeaa
Using the established bound for $\Pi+\frac 12 \trt$ in \eqref{eq:improved-bound-Pi-trt-sum}, we integrate using Lemma \ref{lem:integration-lemma-1} to obtain
\beaa
\|\trt \|_{\H^{s+1}(S_r)}\lesssim r^2 \| \KK+\Kk_{\ell\leq 1}\|_{\H^{s-1}(S_r)} \lesssim \eps_0+ \eps^2.
\eeaa
Then, combining this with \eqref{eq:improved-bound-Pi-trt-sum}, we also deduce the same bound for $\| \Pi \|_{\H^{s+1}(S_r)}$.
\item We integrate  \eqref{eq:curl-Xi-estimate} using Lemma \ref{lem:integration-lemma-1} to bound  $\curlh \Xi$:
\beaa
\| r^3 \curlh \Xi \|_{\H^s(S_r)}\lesssim \eps_0+ \int_{r_0}^r r'^3 \| \KKd\|_{\H^s(S_r)} \lesssim \eps_0 r^2.
\eeaa
 Then, along with the condition $\divh \Xi=0$, we obtain the estimate 
$ \|\Xi \|_{\H^{s+1}(S_r)} \lesssim \eps_0$.
\end{enumerate}
\end{proof}
Therefore, we conclude that $\| \Psi \|_{s,\ga,I} \leq C(\eps_0+ \eps^2 ) \ll \frac 12\eps$, and hence deduce $r^*=\infty$.

\appendix
\section{Technical Lemmas}
\subsection{Basic identities on conformal changes}
\begin{lemma}\lab{lem:div-curl-hot-identities}
Let $\ga=e^{2u}\gz$ on a $2$-sphere, and suppose $\xi\in \ss_1$, $h \in \ss_2$, and $v$ is a vector field. We have the relations between operators defined with respect to $\ga$ and $\gz$:
\begin{align}
\divh \xi=e^{-2u} \divz\xi,\quad \curlh\xi=e^{-2u} \curlz\xi,\quad 
\nabh\hot v^{\flat,\ga} = e^{2u} \nabz\hot v^{\flat,0},\quad \divh h=e^{-2u} \divz h.
\end{align}
\end{lemma}
\begin{proof}

The Christoffel symbol of $\ga$ relative to $\gz$ reads
\beaa
\Gamma(\ga;\gz)_{AB}{}^C=
\delta^{}_A{}^C \nabz_B u+\delta^{}_B{}^C \nabz_A u -\gz_{AB}(\nabz)^C u.
\eeaa
Hence
\beaa
\nabh_A\xi_B=
\nabz_A\xi_B
-(\nabz_A u) \xi_B- (\nabz_B u) \xi_A+\gz_{AB} \nabz u\cdot \xi.
\eeaa
Therefore,
\beaa
\divh \xi= e^{-2u}(\gz)^{AB}
\left(
\nabz_A\xi_B
-(\nabz_A u) \xi_B-(\nabz_B u) \xi_A+\gz_{AB} \nabz u\cdot \xi
\right)
= e^{-2u}\divz \xi.
\eeaa
For the curl, note that
\beaa
\curlh \xi=\ga^{AC}\ga^{BD}\in_{CD}\nabh_A\xi_B.
\eeaa
Since $\in_{CD}=e^{2u}\in\0_{CD}$ and
$\ga^{AC}\ga^{BD}=e^{-4u}(\gz)^{AC} (\gz)^{BD}$, we get
\beaa
\curlh \xi=
 e^{-2u}(\gz)^{AC} (\gz)^{BD}\in\0_{CD}\nabh_A\xi_B=e^{-2u} \curlz\xi.
\eeaa
where we use that $\nabh_A\xi_B-\nabz_A\xi_B$ is symmetric in $A$ and $B$.

We also have
\beaa
(\nabh\widehat\otimes\xi)_{AB}
&=&
\nabh_A\xi_B+\nabh_B\xi_A-\ga_{AB}\divh\xi \\
&=& \nabz_A\xi_B+\nabz_B\xi_A-e^{2u} \gz_{AB} e^{-2u} \divz\xi \\
&& -2(\nabz_A u) \xi_B-2(\nabz_B u)\xi_A+2\gz_{AB} \nabz u\cdot \xi \\
&=&
(\nabz\widehat\otimes\xi)_{AB}
-2(\nabz u\widehat\otimes\xi)_{AB}.
\eeaa
Then, for a vector field $v$, we can compute 
\beaa
\nabh\hot v^{\flat,\ga} &=& \nabz\hot v^{\flat,\ga}-2(\nabz u\hot v^{\flat,\ga}) \\
&=& \nabz\hot (e^{2u} v^{\flat,0})-2e^{2u} (\nabz u\hot v^{\flat,0}) \\
&=& e^{2u} \nabz\hot v^{\flat,0},
\eeaa
where we note that the symmetric traceless product of two $1$-forms is invariant under conformal changes of the metric.

Now let $h$ be a symmetric traceless covariant $2$-tensor. 
Using
\beaa
\nabh_B h_{CA}
=
\nabz_B h_{CA}
-\Ga(\ga;\gz)^{}_{BC}{}^D h_{DA}
-\Ga(\ga;\gz)^{}_{BA}{}^D h_{CD},
\eeaa
we obtain
\beaa
\nabh_B h_{CA}
&=&
\nabz_B h_{CA}
-(\nabz_B u) h_{CA}- (\nabz_C u) h_{BA}
+ \gz_{BC} (\nabz u \cdot h)_{A}\\
&&
-(\nabz_B u) h_{CA}- (\nabz_A u) h_{CB}
+\gz_{BA} (\nabz u\cdot h)_{C}.
\eeaa
Contracting with $\ga^{BC}=e^{-2u}(\gz)^{BC}$, we see that the connection terms again all cancel, which gives
\beaa
\divh h=e^{-2u}\divz h.
\eeaa
\end{proof}

\subsection{Comparison lemma}
\begin{lemma}\lab{lem:comparison}
Consider a sphere $(S,\ga)$, equipped with a spherical coordinate $(\vth^A)$ and the corresponding $\ell=1$ basis $J_p$, with $\ga=e^{2u} \gz=r^2 e^{2u} (d(\vth^1)^2+\sin^2 (\vth^1) d(\vth^2)^2)$ for some $r>0$, where $u$ satisfies the estimate  
\beaa
r^{-1} \| u \|_{\H^{s+1}_0(S)}\leq \mathring \eps \ll 1,
\eeaa
for some $s\geq 3$.
Then the following statement holds:
\begin{enumerate}
\item
For the equation
\bea\lab{eq:2-th-order-equation-h}
\d_1\d_2 h=\sum_{p=0,+,-} (F+c_p J_p+c_0,\dual F+\dual c_p J_p+\dual c_0),
\eea
there exist unique constants $c_0$, $\dual c_0$, $c_p$, $\dual c_p$ for which \eqref{eq:2-th-order-equation-h}  has a unique solution $h\in\ss_2(S,\ga)$.
Moreover, the constants satisfy the estimate
\begin{equation}\lab{eq:estimate-cor-solvability}
\begin{gathered}
 |(c_0,\dual c_0)| \lesssim |\overline{(F,\dual F)}^\ga| ,\\
 |c_p+ \langle F,J_p\rangle_{r^{-2}\gz} |+|\dual c_p+\langle \dual F,J_p\rangle_{r^{-2}\gz}| \lesssim \mathring \eps r^{-1} \|(F,\dual F) \|_{L_\ga^2(S)}.
\end{gathered}
\end{equation}
\item Suppose for scalar field $\phi$, $\xi\in \ss_1$, and $h\in \ss_2(S,\ga)$ we have
\begin{equation*}
\laph_\ga \phi=G,\quad \d_1 \xi=(f,\dual f),\quad \d_2 h=F.
\end{equation*}
Then the following estimates hold for all $i\leq s$:
\begin{equation}\lab{eq:Hodge-estimate-round}
\|\phi-\overline{\phi}^\ga\|_{\H_\ga^{i+2}(S_r)}\lesssim r^2 \| G\|_{\H_\ga^i(S_r)},\quad  \| \xi \|_{\H_\ga^{i+1}(S_r)}\lesssim r \|(f,\dual f) \|_{\H_\ga^i(S_r)},\quad \| h \|_{\H_\ga^{i+1}(S_r)} \lesssim r \| F \|_{\H^i(S_r)}.
\end{equation}
The same estimates hold if all $\H_\ga^s$ norms \eqref{eq:Hodge-estimate-round} are replaced by $\H^s_0$ norms.
\item For $\RR_{ellip}(\phi)$ defined in Definition \ref{def:error-terms}, we have
\beaa
\| \RR_{ellip}(\phi) \|_{\H_\ga^s(S_r)} \lesssim \mathring\eps \|\phi\|_{\H_\ga^s(S_r)} .
\eeaa
\end{enumerate}
\end{lemma}
\begin{proof}
The first statement is Corollary 3.7 in \cite{CK25}. For the second statement, the estimates in \eqref{eq:Hodge-estimate-round} are standard, and the $\H^s_0$ version is obtained using the equivalency of the norms; see Lemma 3.2 of \cite{CK25}.  The third statement follows from a straightforward verification using $\laph=e^{2u}\lapz$.
\end{proof}

\subsection{The commutation formula}

We first recall the commutation formula expressed in geometric quantities defined in Section \ref{sec:geometric-quantities}.
\begin{lemma}\lab{lem:commutation-1}
We have 
\bea\lab{eq:commutation-nab-N-nab-a}
[\nabh_N, \nabh_a] U_{b_1\cdots b_k} = -\th_{ac} \nabh_c U_{b_1\cdots b_k} -\pp_a \nabh_N U_{b_1\cdots b_k} +(p_{b_i} \th_{ac} +p_c \th_{ab_i}+\in_{b_i c} \dual Y_a) U_{b_1\cdots c\cdots b_k}.
\eea
\end{lemma}
\begin{proof}
See Lemma 3.4 in \cite{CK25}.
\end{proof}

We now consider the metric
\bea\lab{eq:widetilde-g}
\widetilde g:= \widetilde a^2 dr^2+ r^2 e^{2 \widetilde u} \gacr_{AB} (d\vth^A +\widetilde b^A dr) (d\vth^B+\widetilde b^B dr).
\eea
We also denote the spherical part of $\widetilde g$ by $\widetilde \ga$. 
As in the main body, we view $\widetilde b$ as a $1$-form obtained by lowering the index of the vector field $\widetilde b^A$ using $\gz$.
Therefore, for the radial normal $\widetilde N$, we have
\bea\lab{eq:expression-N-n-b-n}
\nabz_{\widetilde N} = \widetilde a^{-1} (\nabz_{\pa_r} - \,\widetilde b \cdot \nabz).
\eea
where $\cdot$ means the dot product with respect to $\gz$. We also denote $\widecheck{\widetilde a}:=\widetilde a -1$.

\begin{lemma}\lab{lem:commutation-2}
We have
\begin{equation}\lab{eq:nabz-N-nabz}
[\nabz_{\widetilde N}, \nabz] U =  -\widetilde a^{-1} r^{-1} \nabz U +\nabz (\widetilde a, 
\widetilde b) \cdot (\nabz_{\widetilde N} ,\nabz) U- r^{-2} \widetilde a^{-1} \widetilde b\, U.
\end{equation}
\end{lemma}
\begin{proof}
Using \eqref{eq:expression-N-n-b-n}, we compute, also noting the schematic relation $[\nabz,\nabz]U =r^{-2} U$,
\beaa
[\nabz_{\widetilde N}, \nabz] U &=& [\widetilde a^{-1}\nabz_{\pa_r}, \nabz] U-[\widetilde a^{-1} \,  \widetilde b \cdot \nabz,\nabz] U \\
&=& \widetilde a^{-1} [\nabz_{\pa_r}, \nabz] U +  \widetilde a^{-2} \nabz \widetilde a \cdot \nabz_{\pa_r} U \\
&& +\nabz (\widetilde a^{-1}  \widetilde b)\cdot \nabz U- r^{-2} \widetilde a^{-1} \widetilde b\, U.
\eeaa
Note that $\nabz_{\pa_r}$ can be replaced using $\nabz_{\widetilde N}$ and $\nabz$. Combining the well-known formula, see e.g., equation (3.6) in \cite{CK25},
\beaa
[\nabz_{\pa_r}, \nabz_a] U_{b_1\cdots b_k} = -r^{-1} \nabz_a U_{b_1\cdots b_k},
\eeaa
we obtain the identity \eqref{eq:nabz-N-nabz} as required.
\end{proof}

We need a lemma regarding projecting the transport equations to $\ell\neq 1$ parts. 
\begin{lemma}\lab{lem:RR-transport} 
For $\RR_{tran}(\phi)$ defined in Definition \ref{def:error-terms} (with $N$ being $\widetilde N$ here\footnote{We reserve the tilded notation for quantities arising in the course of the iteration or for their limit.}), we have, schematically
\begin{align*}
r^{-1} \| \RR_{tran} (\phi) \|_{\H_0^s(S_r)} \lesssim r^{-1} \| \widecheck{\widetilde a} \|_{\H^s_0(S_r)}  \| N\phi \|_{L_0^\infty(S_r)}+r^{-1} \| \widetilde a\|_{\H_0^s(S_r)} \| \nabz \phi \cdot\widetilde b\|_{L_0^\infty(S_r)}.
\end{align*}
Here, the $\H_0^s$ norms can be replaced by $\H^s_\ga$, at the expense of allowing the implicit constant to contain a factor $(1+r^{-1} \| u\|_{\H_0^{s-1}(S_r)})$.
\end{lemma}
\begin{proof}
We compute $(\widetilde N \phi)_{\ell=1}-\widetilde N (\phi_{\ell=1})$ for example. We have
\begin{align*}
\widetilde N (\phi_{\ell\neq 1})&= \widetilde N\phi- \widetilde N(\phi_{\ell=1})=N\phi- (\widetilde a^{-1} \pa_r)\Big(\int_{S_r} \phi\, J_p\, d\vol_{\mathbb{S}^2}\Big) \\
&= \widetilde N\phi - \widetilde a^{-1} \int_{S_r} (\pa_r \phi) J_p \, d\vol_{\mathbb{S}^2} =\widetilde N\phi - \widetilde a^{-1} \int_{S_r} (\widetilde a \widetilde N+\widetilde b\cdot\nabz) \phi\, J_p\, d\vol_{\mathbb{S}^2} \\
&= (\widetilde N\phi)_{\ell\neq 1}+\int_{S_r} (\widetilde N\phi) J_p \, d\vol_{\mathbb{S}^2}- \widetilde a^{-1} \int_{S_r} \widetilde a \widetilde N \phi J_p \, d\vol_{\mathbb{S}^2}-\widetilde a^{-1} \int_{S_r} \widetilde b\cdot\nabz \phi J_p \, d\vol_{\mathbb{S}^2}.
\end{align*}
Therefore, using $1 - \widetilde a^{-1} = \widetilde a^{-1}\widecheck{\widetilde a}$, we deduce 
\begin{align*}
\widetilde N (\phi_{\ell\neq 1}) - (\widetilde N\phi)_{\ell\neq 1} = \widetilde a^{-1} \Big[ \widecheck{\widetilde a} (\widetilde N\phi)_{\ell=1} - (\widecheck{\widetilde a} \widetilde N\phi)_{\ell=1} - (\widetilde b\cdot\nabz \phi)_{\ell=1} \Big],
\end{align*}
which yields the estimate  
\begin{align*}
r^{-1} \| \widetilde N (\phi_{\ell\neq 1}) - (\widetilde N\phi)_{\ell\neq 1}\|_{\H_0^s(S_r)} \lesssim r^{-1} \| \widecheck{\widetilde a} \|_{\H^s_0(S_r)}  \| N\phi \|_{L_0^\infty(S_r)}+r^{-1} \| \widetilde a\|_{\H_0^s(S_r)} \| \nabz \phi \cdot\widetilde b\|_{L_0^\infty(S_r)}.
\end{align*}
The estimate for other possible expressions considered in the notation $\RR_{tran}$ is similar.
\end{proof}

\subsection{The integration lemma}
\begin{lemma}\lab{lem:integration-lemma-1}
For a reference metric $\widetilde g$ defined in \eqref{eq:widetilde-g}, consider the equation 
\begin{align}\lab{eq:transport-model-1}
\widetilde N \phi+\la r^{-1} \phi= F,
\end{align} 
with $\phi$ taking value in $\mathbb{R}$ or $\mathbb{R}^2$. Then, we have the estimate
\begin{align}\lab{eq:transport-estimate-gz}
r^{-1} \| r^\la \phi\|_{\H_0^s(S_r)} 
&\lesssim r_0^{-1} \| r_0^\la \phi\|_{\H_0^s(S_{r_0})} 
+ \int_{r_0}^r \Big( \|r'^{\la-1} F\|_{\H_0^s(S_{r'})} + \RR_{com}^s(\phi)(r') \Big) \, dr',
\end{align}
where the error $\RR_{com}^s(\phi)$ satisfies
\begin{align}\lab{eq:transport-error-bound}
\RR_{com}^s(\phi) \lesssim \| (r^{-1} \widecheck{\widetilde a}, \nabz \widetilde b) \|_{L^\infty(S_r)} \| r^{\la-1} \phi\|_{\H_0^s(S_r)} + \| (r^{-1} \widecheck{\widetilde a}, \nabz \widetilde b) \|_{\H^s_0(S_r)} \| r^{\la-1} \phi \|_{L^\infty(S_r)}.
\end{align}
Now consider the equation 
\begin{align}\lab{eq:transport-model-2}
\widetilde N \phi+\la r^{-1} \phi=r^{-1} A\phi+ F,
\end{align} 
where $A$ is either a scalar operator or, in the $\mathbb R^2$-valued case, a $2\times 2$ matrix of operators on $S_r$. Assume that each operator entry of $A$ is constant multiplication or $\laph^{-1}$, and $A$ is non-positive on $\H^i_{\widetilde\ga}(S_r)$, respectively on $\H^i_{\widetilde\ga}(S_r)\times \H^i_{\widetilde\ga}(S_r)$ equipped with the inner product
\[
\langle\cdot,\cdot\rangle_{\H^i_{\widetilde\ga}}\otimes G,
\]
for every nonnegative integer $i$, where $G$ is a fixed positive definite inner product on $\mathbb R^2$.\footnote{In practice, $A$ will be taken to be the matrix operator appearing in \eqref{eq:N-de-weight-coupled}.}
Then, we have the estimate for $\H_{\widetilde \ga}^s$ norms defined using $\widetilde\ga$:
\begin{align}\lab{eq:transport-estimate-ga}
r^{-1} \| r^\la \phi\|_{\H_{\widetilde \ga}^s(S_r)} 
&\lesssim r_0^{-1} \| r_0^\la \phi\|_{\H_{\widetilde \ga}^s(S_{r_0})} 
+ \int_{r_0}^r \Big( \|r'^{\la-1} F\|_{\H_{\widetilde \ga}^s(S_{r'})} + \RR_{com}^s(\widetilde g, \phi)(r') \Big) \, dr',
\end{align}
where the error $\RR_{com}^s(\widetilde g, \phi)$ satisfies
\begin{align}\lab{eq:transport-error-bound-sharp}
\RR_{com}^s(\widetilde g, \phi) \lesssim \|\Gamma_1(\widetilde g)\|_{L^\infty} \| r^{\la-1} \phi\|_{\H_{\widetilde \ga}^s(S_r)} + \|\Gamma_1(\widetilde g)\|_{\H_{\widetilde \ga}^{s-1}(S_r)} \| r^{\la-1} (r\nabz \phi,\phi) \|_{L^\infty}.
\end{align}
Here, $\Ga_1(\widetilde g)$ is the geometric quantities in \eqref{eq:schematic-quantities}, defined through $\widetilde g$. 
\end{lemma}
\begin{proof}
We first prove \eqref{eq:transport-estimate-gz}.
We have, for any scalar field $f$,
\begin{align*}
\pa_r  \int_{S_r} f\, d\vol_{\gz} &= \int_{S_r} (\pa_r f +2r^{-1} f)\, d\vol_{\gz} = \int_{S_r} (\widetilde a \widetilde N f + \widetilde b\cdot\nabz f+2r^{-1} f)\, d\vol_{\gz} \\
&= \int_{S_r} \left(\widetilde a \widetilde N f+2r^{-1} f-(\divz \widetilde b) f\right) d\vol_{\gz}.
\end{align*}
Applying it to $f=|r^{\la-1} \phi|^2$, using that $\widetilde N(r)=\widetilde a^{-1}$, we compute
\begin{align*}
&\quad \pa_r \int_{S_r}  |r^{\la-1} \phi|^2  = \int_{S_r}  2\,  \widetilde a \, r^{\la-1} \phi  \cdot \widetilde N (r^{\la-1} \phi) +2r^{-1} |r^{\la-1} \phi |^2 - (\divz \widetilde b)|r^{\la-1} \phi |^2  \\
&= \int_{S_r}  2\, \widetilde a\, r^{\la-1} \phi \cdot \left((\la-1) \widetilde N(r) r^{\la-2} \phi+r^{\la-1} \widetilde N \phi\right) +2r^{-1} |r^{\la-1} \phi |^2 - (\divz \widetilde b)|r^{\la-1} \phi|^2  \\
&= \int_{S_r} 2 r^{2\la-3} (\la-1) |\phi|^2 +2\, \widetilde a\, r^{2\la-2} \phi\cdot (F-\la r^{-1} \phi) +2r^{-1} |r^{\la-1} \phi |^2 - (\divz \widetilde b)|r^{\la-1} \phi|^2. 
\end{align*}
Hence, schematically,
\beaa
\pa_r \int_{S_r}  |r^{\la-1} \phi|^2  = \int_{S_r} r^{2\la-2} (r^{-1}\widecheck{\widetilde a},\nabz\, \widetilde b) |\phi|^2+2\, \widetilde a\, r^{2\la-2} F \cdot \phi.
\eeaa
The estimate \eqref{eq:transport-estimate-gz}, in the case when $s=0$, follows from H\"older's inequality. The higher-order estimates follow by commuting \eqref{eq:transport-model-1} with $\nabz$ using Lemma \ref{lem:commutation-2} iteratively:
\begin{align*}
\nabz_{\widetilde N} (\nabz)^s \phi + (\la+s) r^{-1} (\nabz)^s \phi &= (\nabz)^s F +\sum_{i+j=s, i\leq s-1} (\nabz)^i (\nabz (\widetilde a,\widetilde b,\widetilde u))\cdot (\nabz,\nabz_{\widetilde N})^j \phi.
\end{align*}
We now prove \eqref{eq:transport-estimate-ga}.
In what follows, the inner product $\cdot$, if $\phi$ takes value in $\mathbb{R}^2$, is taken with respect to $G$ where $A$ verifies non-positivity.
We recall the more standard formula
\beaa
&& \pa_r \int_{S_r} f \, d\vol_{\widetilde \ga}=\int_{S_r} \widetilde a (\widetilde N f+ \trth(\widetilde g) f) \, d\vol_{\widetilde \ga},
\eeaa
where we can write $\trth(\widetilde g)=\thc(\widetilde g)+2r^{-1}$.
Then, through a similar calculation, this time taking into account the contribution of $A$, we deduce, schematically for the right-hand side,
\beaa
\pa_r \int_{S_r}  |r^{\la-1} \phi|^2 - \int_{S_r} 2 \widetilde a r^{2\la-3}  A\phi\cdot \phi = \int_{S_r} r^{2\la-2}  \Ga_1(\widetilde g) |\phi|^2+2\, \widetilde a\, r^{2\la-2} F \cdot \phi,
\eeaa
where we note that $r^{-1} \widecheck{\widetilde a}$ is included in $ \Ga_1(\widetilde g)$. Then, using the sign of $A$ and H\"older's inequality, we conclude the proof of the $s=0$ case. The higher-order case follows from Lemma \ref{lem:commutation-1} and the schematic commutation formula $[\nabh^s, \laph^{-1}]U=O( r^{-2} \nabh^{s-4} U)$; Note that the inverse Laplacian $\laph^{-1}$ is defined on covariant tensors by $\laph^{-1} U:=V$, where $V\in \mathrm{Ker}(\laph)^\perp$ is the unique solution to $\laph V=P^\perp U$, with $P^\perp $ denoting the orthogonal projection onto $\mathrm{Ker}(\laph)^\perp$ in $L^2(S)$. Alternatively, one can commute with the multiplier $|\nabh|$.
\end{proof}

\end{document}